\documentclass[11pt,english]{amsart}
\usepackage[T1]{fontenc}
\usepackage[latin9]{inputenc}
\usepackage{srcltx}
\usepackage{amstext}
\usepackage{amsthm}
\usepackage{amssymb}
\usepackage{setspace}
\usepackage[authoryear]{natbib}
\makeatletter

\providecommand{\tabularnewline}{\\}
\usepackage{tikz}
\usetikzlibrary{calc}

\numberwithin{equation}{section}
\numberwithin{figure}{section}
\theoremstyle{plain}
\newtheorem{thm}{\protect\theoremname}
\theoremstyle{definition}
\newtheorem{example}[thm]{\protect\examplename}
\theoremstyle{remark}
\newtheorem{rem}[thm]{\protect\remarkname}
\theoremstyle{plain}
\newtheorem{lem}[thm]{\protect\lemmaname}

\usepackage{accents}
\usepackage{setspace}
\usepackage{graphicx}
\usepackage{enumitem}

\usepackage{units}
\usepackage{amsfonts}
\usepackage{amsthm}
\usepackage{amsxtra}
\usepackage{subfigure}
\usepackage{bbm}
\usepackage{amsaddr}
\usepackage[english]{babel}
\usepackage{xcolor}

\numberwithin{equation}{section}
\numberwithin{figure}{section}

\theoremstyle{plain}
\newtheorem{assumption}{Assumption}

\newcommand{\PSI}{\mathbf{\Psi}}

\allowdisplaybreaks

\makeatother

\usepackage{babel}
\providecommand{\examplename}{Example}
\providecommand{\lemmaname}{Lemma}
\providecommand{\remarkname}{Remark}
\providecommand{\theoremname}{Theorem}

\begin{document}

\title[Nonparametric for Economic Experiments]{Nonparametric Estimation and Inference in Economic and Psychological
Experiments}

\author{Raffaello Seri, Samuele Centorrino, Michele Bernasconi}

\date{Last Version: \today; First Version September 1, 2013}

\address{Università dell'Insubria, Stony Brook University, Università Ca'
Foscari}
\begin{abstract}
The goal of this paper is to provide some tools for nonparametric
estimation and inference in psychological and economic experiments.
We consider an experimental framework in which each of $n$subjects
provides $T$ responses to a vector of $T$ stimuli. We propose to
estimate the unknown function $f$ linking stimuli to responses through
a nonparametric sieve estimator. We give conditions for consistency
when either $n$ or $T$or both diverge. The rate of convergence depends
upon the error covariance structure, that is allowed to differ across
subjects. With these results we derive the optimal divergence rate
of the dimension of the sieve basis with both $n$ and $T$. We provide
guidance about the optimal balance between the number of subjects
and questions in a laboratory experiment and argue that a large $n$ is
often better than a large $T$. We derive conditions for asymptotic
normality of functionals of the estimator of $T$and apply them to
obtain the asymptotic distribution of the Wald test when the number
of constraints under the null is finite and when it diverges along
with other asymptotic parameters. Lastly, we investigate the previous
properties when the conditional covariance matrix is replaced by an
estimator.
\end{abstract}

\thanks{This work was supported by a grant from MIUR (Ministero dell'Istruzione,
dell'Università e della Ricerca). Authors would like to thank seminar
participants at Stony Brook University, University of California Riverside
and the 2016 UK Econometrics Study Group, especially Xavier D'Haultfoeuille,
for useful comments and remarks.}
\maketitle

\section{Introduction}

The aim of this paper is to provide a statistical theory useful for
the nonparametric analysis of laboratory experiments in economics
and psychology.

In the typical experiment we have in mind, there are $n$ subjects
who are administered $T$ tasks. Task $t$ is characterized by $X_{t}$,
a $d$-dimensional stimuli-vector that is the same for each subject
$i$, for $i=1,\dots,n$. The response or choice of subject $i$ in
task $t$ is denoted by $Y_{it}$. We suppose that the \emph{Data
Generating Process} (DGP) of the set of answers $Y_{it}$ ($i=1,...,n$,
$t=1,...,T$) has a deterministic component represented by a nonparametric
function $f\left(X_{t}\right)$ of the stimuli, and a stochastic component
$\varepsilon_{it}$ arising from individual error terms. Hence, we
have the system: 
\begin{equation}
Y_{it}=f(X_{t})+\varepsilon_{it},\text{ \ \ \ \ \ \ }i=1,\dots,n\text{ \ \ }t=1,\dots,T.\label{EQ:Basic}
\end{equation}
The function $f:\mathcal{X}\mapsto\mathbb{R}$ can be interpreted
as a deterministic theory that maps every value of the vector of stimuli
$X$ to a space of real valued responses.

Equation \eqref{EQ:Basic} resembles a framework already discussed
in experimental economics by \citet{hey2005}, which involves a nonstandard
econometric system, and encompasses several models arising in economics
and psychologica experiments. To the best of our knowledge, however,
a complete statistical analysis of this system has yet to be conducted.
For instance, the vector $X_{t}$ can represent the prizes and the
probabilities of a lottery presented in an economic experiment in
which subjects are asked to give the certainty equivalent of the lottery.
In this case, $f\left(\cdot\right)$ is the function used by subjects
to evaluate the lotteries.\footnote{The present approach can be equally applied, though it may not be
the most efficient, to situations in which $X_{t}$ is defined as
$X_{t}=\left(a_{t},b_{t}\right)$, for two lotteries $a_{t}$ and
$b_{t}$ presented to subjects in pairwise choice experiments, and
$Y_{it}$ is simply the choice (coded in some way) of subject $i$.} In psychophysical experiments, the vector $X_{t}$ can be thought
to represent stimuli such as light, sound, weight, distance for which
subjects are asked to assess in pairwise comparisons the relative
magnitude. In this case, $f\left(\cdot\right)$ is the scale used
by subjects to measure the stimuli. The fact that the explanatory
variables are the same across individuals has a double justification:
first of all, in many empirical studies the choice of the stimuli
is so difficult that it is not possible to conduct it for each individual;
second, when the regressors are the same across individuals the estimation
problem is more difficult (in the sense that we have less information
from the variation in the independent variables). 

In statistics and econometrics, this model can be cast in the well-known
and extensively studied framework of the nonparametric regression
model \citep{liracine2007,tsybakov2009}. There are, however, several
distinctive features of this model that make its analysis different
and, in some respects, more challenging than the standard nonparametric
regression. First, the realizations of the stimuli are not random
observations from an underlying statistical distribution, but are
chosen by the experimenter: the more complex the function one wishes
to estimate, the richer the support of the data one needs to achieve
consistency. Second, the statistical approach proposed in this work
does not impose any specific restrictions on the structure of the
error terms. In particular, it seems important that even if \eqref{EQ:Basic}
holds true, the error terms for different individuals should be allowed
to be differently distributed, provided $\mathbb{E}\left(\varepsilon_{it}\right)=0$.
Among other things, this means that we allow different individuals
to have different degrees of precision. This seems especially important
in economics and psychology, when a researcher may have little to
no knowledge about a theory that explains randomness in the responses.
Also, individual variances may contain very persistent components,
and, therefore, consistent estimation of an individual-specific response
function may be unfeasible.\footnote{A similar statistical framework has been studied by \citet{stanis1998}.
While they also allow for the stimuli-vector to be time-varying, they
suppose that the error term is a white noise. }

Our goal in this paper is to study the nonparametric estimation of
the function $f\left(\cdot\right)$ in \eqref{EQ:Basic} using the
method of sieves \citep[see][among others]{newey1997,jong2002,chen2007,belloni2014,chenchristensen2013}.
The function of interest is approximated by a finite linear combination
of some known basis functions (e.g., power series, regression splines,
trigonometric polynomials), which effectively reduce the estimation
problem to a finite number of parameters. The weights in the function
approximation can be estimated through linear regression supposing
that individuals and answers across individuals are independent. The
number of approximating terms diverges to infinity with the sample
size. We show that this estimator of the function $f\left(\cdot\right)$
is consistent, and we provide the convergence rate for this nonparametric
estimator. We show that the convergence rate depends on the number
of tasks ($T$), the number of individuals ($n$), and the number
of basis functions used to approximate $f\left(\cdot\right)$. Our
convergence rate, however, also depends on the properties of the average
covariance matrix across individuals for a given task $t$. Heuristically,
it thus explicitly takes into account the precision of the subjects
in answering the questions and/or selecting specific choices.

We also provide asymptotic normality results for both linear and nonlinear
functionals of the nonparametric estimator which are useful to obtain
the asymptotic properties of the Wald test in this framework \citep{chen2012b}.
We derive the properties of the latter both in the case where the
number of constraints is finite (parametric restrictions), which gives
the standard $\chi^{2}$ distribution under the null; and when the
number of constraints diverges to infinity along with other asymptotic
parameters (a normal distribution). Lastly, we investigate what happens
when the average variance matrix appearing in the previous tests is
replaced by an estimator. We believe inference is an essential part
of our statistical theory, as it allows us to test specific behavioral
assumptions.

\citet{hey2005} points out that underlying system \eqref{EQ:Basic}
is the idea that the theory under investigation is deterministic,
but that people apply the theory with noise. Such an approach, which
is sometimes referred to as Thurstonian or Fechnerian, underlies for
example the investigations conducted by \citet{falmagne1976}, \citet{orme1994},
\citet{buschena2000}, and \citet{blavatskyy2007}. Alternatively,
other authors (including \citealp{camerer1994}, \citealp{loomes1995},
\citealp{loomes2002}, \citealp{myung2005}) tend to interpret individual
behavior in experiments (and possibly in the real world) as inherently
stochastic, in the sense that while the theories remain deterministic,
their predictions are not because of the imprecision of people to
know and to use the same specification of the theory every time it
is required.\footnote{For example, describing the philosophy behind the approach with reference
to preference theories, \citet{loomes2005} argues that the approach
``rather than supposing an individual to have a single true preference
function to which white noise is added, ... treats imprecision \emph{as
if} an individual's preference consist of a \emph{set} of such functions.
Thus to say that a particular individual behaves according to a certain
`core' theory is to say that the individuals' preferences can be represented
by some functions, all of which are consistent with the theory; but
that on any particular occasion, the individual act \emph{as if} she
picks one of those functions at random from the set, applies it to
the decision at hand, then replaces that function before picking again
at random in order to address the next decision'' (p. 306). Antecedents
of this approach can be found in \citet{becker1963}. It is also important
to emphasize that this approach is still very different from the case
in which the `core' theory itself would be made inherently stochastic
--- as for example advocated by \citet{luce1997}.} The distinction between the two approaches, however, though quite
interesting philosophically, is of practical relevance only when either
of the following two circumstances applies. The first is when the
dependence of the answers $Y_{it}$ from the stimuli $X_{t}$ (according
to function $f\left(\cdot\right)$) is parametric. In this case, the
question of the two approaches turns into a fundamental question about
whether the parameters to be estimated can be interpreted in deterministic
terms or as random variables (as for example in the Bayesian approach
pursued by \citealp{karabatsos2005}, and \citealp{myung2005}). The
second applies to the specific restrictions on the errors terms which
may be required by the statistical procedure used to analyze the experimental
data (see for example \citealp{ballinger1997}, for the discussions
of several restrictions often imposed for the empirical analyses of
data from decision theory experiments).

The present statistical approach is unaffected by both circumstances,
so that it can be viewed to encompass both philosophies. First of
all, a notable feature of the present approach is that the dependence
of the answers $Y_{it}$ from the stimuli $X_{t}$ is left nonparametric.
Nonparametric dependence has the advantage that theoretical and/or
behavioral properties of interest can be estimated and tested without
the mediation of parametric restrictions, which (for the reasons just
exposed) may not be unambiguously interpreted. Furthermore, nonparametric
dependence is natural if one wishes to fit the experimental data without
imposing any restrictions on behavior. Such `unrestricted' model could
be a useful benchmark against which to compare any structural model
indicated by specific theories.\footnote{See \citet{bernasconi2008,bernasconi2010a}, for applications of such
an approach in regards to experimental investigations of, respectively,
psychophysical measurement theories and decisions theories.} 

As mentioned above, we also allow for the possibility that the precision
of answers of the same individual varies across different questions.
This is important because various previous studies have emphasized
forms of heterogeneity occurring both at levels of individuals and
of different experimental tasks \citep[e.g.,][]{ballinger1997,buschena2000,carbone2000,blavatskyy2007,butler2007}.

Finally, we should note that a long debated dispute in economic and
psychological experiments is whether the analysis of the individual
responses $Y_{it}$ should be conducted for the aggregate of the individuals
or individual by individual. The analysis presented here is primarily
thought for the former case. In particular, depending on the degree
of heterogeneity and precision of the experimental sample and of the
theory one would like to test, large values of $n$ and/or of $T$
may have different impacts on the consistency of our estimator of
the function $f\left(\cdot\right)$ and its derivatives.\footnote{An additional reason to prefer an aggregate analysis is that an experimenter
may decide to assign different values of the stimuli to different
subjects. Assuming that the assignment mechanism is random, the aggregate
model would allow to approximate the function $f\left(\cdot\right)$
on a richer support. In this situation, the convergence rates presented
in this paper constitute a worst-case scenario.} If, however, one believes that the aggregate analysis cannot be carried
forward because all individuals are characterized by different functions
$f_{i}\left(X_{t}\right)$,\footnote{In psychology, the risks of averaging across individuals when they
are characterized by different functions has been stressed, e.g.,
in \citet[p. 99]{skinner_reinforcement_1958}, \citet[p. 212]{yost_variability_1981}
and \citet{Bernasconi-Seri-JMP-2016}.} then our results apply verbatim, simply taking $n=1$ and letting
the number of tasks $T$ to diverge. In this case, our analysis can
be seen as an extension of the results in \citet{newey1997} and \citet{belloni2014}
to the case of deterministic regressors. Consistency is then guaranteed
only under more stringent conditions on the variance of each individual
error term.

The model in \eqref{EQ:Basic} can also be interpreted as a panel
data specification in which the covariates vary only with $t$. We
do not pursue this interpretation further in this work, but we notice
that \citet{su2012} have considered a panel data model with factor
structure in the error term in which the function $f\left(\cdot\right)$
is allowed to vary across individuals, with $n,T\rightarrow\infty$.
Notice, however, that time series data have a natural ordering that
can be used in the asymptotic analysis, while our data do not possess
such ordering. We show below that our convergence rates are amenable
to some of their results, upon additional restrictions on the variance
of the error term.

The present statistical approach is also suitable for various extensions
which we indicate in the conclusions.

\section{The Statistical Model\label{Sect - Model}}

We recall that the data generating process is modeled as follows:
\begin{equation}
Y_{it}=f(X_{t})+\varepsilon_{it},\label{EQ:Basic2}
\end{equation}
where $i=1,\dots,n$ denotes the individual (or respondent) and $t=1,\dots,T$
denotes a specific task. The dependent variable $Y\in\mathbb{R}$
and the vector of independent variables $X\in\mathcal{X}$, where
$\mathcal{X}$ is taken to be a compact subset of $\mathbb{R}^{d}$.\footnote{Taking $\mathcal{X}$ to be compact does not appear to be a strong
restriction in this setting, as points are chosen by the experimenter
possibly within a bounded interval. In the development of our theory,
this assumption could be relaxed by substantially modifying our method
of proof \citep[see][]{chenchristensen2013}.} In the following, we will suppose that the function $f\left(\cdot\right)$
belongs to a space $\mathcal{F}$ that will not be specified explicitly:
when in the discussion of the results we will suppose that $f\left(\cdot\right)$
has at least $s$ continuous derivatives, it is intended that $\mathcal{F}$
will coincide with the Sobolev space $\mathcal{W}_{s,\infty}=\left\{ f:\left|f\right|_{s}<\infty\right\} $.

This statistical models fits several experimental and quasi-experimental
frameworks.
\begin{example}
{[}Cumulative Prospect Theory{]}\label{Ex - CPT - 1} In the following,
$t$ denotes the given lottery and $i$ the individual. Consider a
gamble $\left(X_{t},p_{t},0\right)$. Let $CE_{it}$ be the \emph{certainty
equivalent} that the individual associates with the gamble, i.e. the
certain monetary amount that makes her indifferent between the two.
Cumulative Prospect Theory (CPT) starts from the following model:
\[
U_{1}\left(CE_{it}\right)=U_{2}\left(X_{t}\right)\cdot g\left(p_{t}\right)+U_{2}\left(0\right)\cdot\left[1-g\left(p_{t}\right)\right],
\]
where $U_{1}$ is the utility function, $U_{2}$ is called the \textit{evaluation
functional}, which may or may not be equal to $U_{1}$, and $g$ is
the probability weighting function, that is taken to be strictly monotone
increasing in $p_{t}$. Supposing that $U_{2}\left(0\right)=0$, we
get:
\[
U_{1}\left(CE_{it}\right)=U_{2}\left(X_{t}\right)\cdot g\left(p_{t}\right).
\]
Now, the experimental elicitation of the certainty equivalent has
attracted critiques for its unreliability \citep[see, e.g.,][Section 1.2.2.1]{hershey1985probability,wakker1996eliciting,Harrison-Rutsroem-BOOK-2008,luce1999utility}.
Several authors have advocated instead to elicit the probability equivalent
$PE_{it}$ such that the individual is indifferent between $\left(X_{t},PE_{it},0\right)$
and $CE_{t}$. The previous relation can be written as:
\[
g\left(PE_{it}\right)=\frac{U_{1}\left(CE_{t}\right)}{U_{2}\left(X_{t}\right)}.
\]
We can write:
\[
PE_{it}=g^{-1}\left(U_{1}\left(CE_{t}\right)/U_{2}\left(X_{t}\right)\right)=F\left(V_{1}\left(\ln CE_{t}\right)-V_{2}\left(\ln X_{t}\right)\right).
\]
where $V_{j}\left(\cdot\right)=\ln U_{j}\left(\exp\left(\cdot\right)\right)$,
for $j=1,2$, and $F\left(\cdot\right)=g^{-1}\left[\exp\left(\cdot\right)\right]$.
Finally, supposing that the previous representation holds with some
error $\varepsilon_{it}$, this implies: 
\[
PE_{it}=F\left(V_{1}\left(\ln CE_{t}\right)-V_{2}\left(\ln X_{t}\right)\right)+\varepsilon_{it}.
\]
\begin{example}
{[}Stevens' Model{]}\label{Ex - Stevens Model - 1} In \emph{ratio
magnitude estimation}, one of the most common form of psychophysical
experiments, the aim is to evaluate the intensity of a set of stimuli
with respect to a reference stimulus whose intensity is set to $1$,
thus justifying the alternative name of magnitude estimation with
a standard \citep[see, e.g.,][]{luce1988}. In task $t$ of the experiment,
an individual $i$ is proposed two stimuli, $X_{1t}$ and $X_{2t}$,
and asked to state the ratio $p_{t}$ of their intensities. One of
the most well-known models in mathematical psychology is Stevens'
model, in which (see \citealp{stevens_psychophysics:_1975,kornbrot_human_2014,Bernasconi-Seri-JMP-2016}):
\[
p_{it}=\left(\frac{X_{1t}}{X_{2t}}\right)^{\kappa}.
\]
It is generally, but not always, the case that $X_{1t}>X_{2t}$ and
$p_{it}>1$. Taking logarithms, we get $\ln p_{it}=\kappa\ln\left(X_{1t}/X_{2t}\right)$.
In order to estimate the model, we set $Y_{it}\triangleq\ln p_{it}$
and $X_{t}\triangleq\ln\left(X_{1t}/X_{2t}\right)$ to get a regression
model without intercept.
\end{example}

\end{example}

We now rewrite the model in equation \eqref{EQ:Basic2} using matrix
notations. We form the $T\times1$-vectors $Y_{i}=\left[Y_{i1},\dots,Y_{iT}\right]^{\prime}$
and $\varepsilon_{i}=\left[\varepsilon_{i1},\dots,\varepsilon_{iT}\right]^{\prime}$.
We suppose that $\varepsilon_{i}$ has a distribution with mean $0$
and variance $\Sigma_{i}$, for every $i=1,\dots,n$. We further define
the $\left(T\times d\right)$-matrix $\mathbf{X}$ obtained stacking
the vectors $\lbrace X_{t},t=1,\dots,T\rbrace$. Finally, 
\[
Y_{i}=f\left(\mathbf{X}\right)+\varepsilon_{i},\quad\quad i=1,\dots,n,
\]
where the function $f\left(\cdot\right)$ is supposed to apply row-wise
to the matrix $\mathbf{X}$. We make the following assumption about
the vector of errors $\varepsilon_{i}$. 

\begin{assumption}\label{ASS1}~ 
\begin{itemize}
\item[(i)] The random vector $\varepsilon_{i}$ is such that $\mathbb{E}\left(\varepsilon_{i}\right)=0$,
$\mathbb{E}\left(\varepsilon_{i}\varepsilon_{i}^{\prime}\right)=\Sigma_{i}$
for all $i=1,\dots,n$, and $\mathbb{E}\left(\varepsilon_{i}\varepsilon_{i^{\prime}}^{\prime}\right)=0_{T\times T}$,
for all $i,i^{\prime}=1,\dots,n$ and $i\neq i^{\prime}$. 
\item[(ii)] Every element of $\Sigma_{i}$ is finite and the matrix $\Sigma_{i}$
is positive definite for all $i=1,\dots,n$. 
\end{itemize}
\end{assumption}

There are not noteworthy details in this assumption. We take the error
terms to be uncorrelated across individuals $i$ and we impose some
regularity conditions on the covariance matrix, which is otherwise
left unspecified.

We now structure the statistical model for the whole data. We build
the $\left(nT\times1\right)-$vectors $\mathbf{Y}$ and $\varepsilon$
by stacking respectively the $n$ vectors $Y_{1},\dots,Y_{n}$ and
$\varepsilon_{1},\dots,\varepsilon_{n}$.

We finally have: 
\begin{equation}
\mathbf{Y}=e\otimes f\left(\mathbf{X}\right)+\varepsilon,\label{Eq - Mean Model}
\end{equation}
where $e$ is a $\left(n\times1\right)$-vector of ones. $\varepsilon$
is then a vector with mean $0$ and variance $\bigoplus_{i=1}^{n}\Sigma_{i}$,
where $\bigoplus$ denotes the direct sum of matrices. That is, $\bigoplus_{i=1}^{n}\Sigma_{i}=\textrm{diag}\left(\Sigma_{1},\dots,\Sigma_{n}\right)$.
\begin{rem}
\label{Ex - Heterogeneity in Variances} The fact that every individual
is allowed to have a potentially different covariance matrix is a
crucial characteristic in our setting. Consider a random function
$f:\mathcal{X}\times\Omega\rightarrow\mathbb{R}$, and suppose that
the decision model for individual $i$ is defined by the function
$f_{i}\left(\cdot\right)=f\left(\cdot,\omega_{i}\right)$, where $\omega_{i}$
is a drawing from a random variable $\omega$, and denote $f\left(\cdot\right)\triangleq\mathbb{E}_{\omega}f\left(\cdot,\omega\right)$,
so that: 
\begin{align*}
Y_{it}= & f_{i}\left(X_{t}\right)+\eta_{it}=f\left(X_{t},\omega_{i}\right)+\eta_{it}\\
= & \mathbb{E}_{\omega}f\left(X_{t},\omega_{i}\right)+f_{i}\left(X_{t}\right)-\mathbb{E}_{\omega}f\left(X_{t},\omega_{i}\right)+\eta_{it}\\
= & f\left(X_{t}\right)+\left\{ f_{i}\left(X_{t}\right)-f\left(X_{t}\right)+\eta_{it}\right\} .
\end{align*}
Here the average $f$ is independent of the individual, but the error
term is heteroskedastic (in the sense that it depends on the regressors)
and heterogeneous (in the sense that it is different across individuals).
In this setting, part of the correlation in the residuals is induced
by the averaging across individuals, and conducting the analysis at
the level of the single individual may improve inferences. 
\end{rem}

For estimation and inference, we take an approximation of $f\left(\cdot\right)$
using a linear combination of basis functions in $\mathcal{X}$. Thus,
at $X_{t}=x$, we take 
\[
f_{P}\left(x\right)=\psi_{P}\left(x\right)\beta,
\]
where $\psi_{P}\left(x\right)=\left[\psi_{1,P}\left(x\right),\dots,\psi_{P,P}\left(x\right)\right]$
is a $1\times P$ vector of given basis functions and $\beta$ a $P\times1$
vector of unknown coefficients, with $P\rightarrow\infty$, with $n,T$.
When $d$, the dimension of the vector of stimuli, is larger than
$1$, then can be taken as a tensor product basis of total dimension
$P$. We also denote as 
\[
\PSI=\begin{bmatrix}\psi_{P}\left(X_{1}\right)\\
\vdots\\
\psi_{P}\left(X_{T}\right)
\end{bmatrix}.
\]
the $\left(T\times P\right)$-matrix that stacks the approximating
bases at every point $\lbrace X_{t},t=1,\dots,T\rbrace$. Then, we
finally have: 
\begin{equation}
\mathbf{Y}=e\otimes\PSI\beta+U,\label{Eq - Vector Form of Polynomial Regression}
\end{equation}
where 
\[
U=e\otimes\left(f(\mathbf{X})-\PSI\beta\right)+\varepsilon.
\]
The true value of the parameter $\beta$, which we denote $\beta_{0}$
is taken to satisfy 
\[
\mathbb{E}_{i}\left[\left(e\otimes\PSI\right)^{\prime}\left(\mathbf{Y}-e\otimes\PSI\beta_{0}\right)\right]=0,
\]
where the expectation is taken with respect to the distribution of
$Y$ for each individual $i$. That is, 
\[
\beta_{0}=e\otimes\left[\PSI^{\prime}\PSI\right]^{-1}\PSI^{\prime}\mathbb{E}_{i}\left[\mathbf{Y}\right].
\]
The estimation of the model is performed by least squares, under the
hypothesis that $\mathbb{E}(U)=0$. Our sieve-based
least-squares estimator is therefore given by: 
\begin{align*}
\widehat{\beta}= & \left[\left(e\otimes\PSI\right)^{\prime}\left(e\otimes\PSI\right)\right]^{-1}\left(e\otimes\PSI\right)^{\prime}\mathbf{Y}\\
= & \frac{1}{n}\left\{ e^{\prime}\otimes\left[\left(\PSI^{\prime}\PSI\right)^{-1}\PSI^{\prime}\right]\right\} \mathbf{Y},
\end{align*}
where $\otimes$ denotes a matrix Kronecker product and we denote
\[
f_{P,0}(\mathbf{X})\triangleq\PSI\beta_{0},\text{ and }\hat{f}_{P}(\mathbf{X})\triangleq\PSI\hat{\beta}.
\]
In some instances, we may omit the argument of the function, and simply
use the notations $\hat{f}_{P}$ and $f_{P,0}$. Notice that the estimator
(and the model itself) could be simply written as 
\[
\hat{\beta}=\left(\PSI^{\prime}\PSI\right)^{-1}\PSI^{\prime}\bar{\mathbf{Y}}
\]

where $\bar{\mathbf{Y}}$ is a $T\times1$vector of average individual
responses for a given task $t=1,\dots,T$. Our model could be then
treated as any other nonparametric regression model with correlated
errors. However, we believe this interpretation defies the entire
purpose of our analysis. As a matter of fact, we would like emphasize
the idea that the number of individuals performing the same task is
essential when we cannot make standard regularity assumptions about
the error term.

\section{Consistency and Convergence Rates\label{Sect - Consistency}}

We first define the norms: 
\begin{align*}
\left|f\right|_{s}\triangleq & \max_{\left|\lambda\right|\le s}\sup_{x\in\mathcal{X}}\left|\partial^{\lambda}f\left(x\right)\right|,\\
\left\Vert f\right\Vert _{\infty}\triangleq & \sup_{x\in\mathcal{X}}\left|f\left(x\right)\right|.
\end{align*}
The norm $\left|\cdot\right|_{s}$ is what is sometimes written as
$\left\Vert \cdot\right\Vert _{s,\infty}$. We do not need to consider
more general weighted Sobolev norms \citep[see, e.g.,][for a definition]{gallant1987},
since the functions we consider have nonstochastic arguments.

Throughout the paper we will also need the following quantities. For
every $s\geq0$:
\[
\zeta_{s}\left(P\right)\triangleq\max_{\left|\lambda\right|\le s}\sup_{x\in\mathcal{X}}\left\Vert \partial^{\lambda}\psi_{P}\left(x\right)\right\Vert _{F},
\]
where $\Vert\cdot\Vert_{F}$ denotes the Frobenius norm. For every
integer $s\geq0$, we define: 
\[
N_{P}\triangleq\left|f-f_{P,0}\right|_{s}.
\]
Denote $\lambda_{nT}\triangleq\lambda_{\max}\left(\bar{\Sigma}\right)$
and $\tau_{nT}\triangleq\mathrm{tr}\left(\bar{\Sigma}\right)$ to
be the largest eigenvalue and the trace, respectively, of the average
covariance matrix of the errors $\bar{\Sigma}\triangleq\frac{1}{n}\sum_{i=1}^{n}\Sigma_{i}$.

The following assumption is needed to derive, together with the previous
definitions, a uniform upper bound for the sieve estimator $\hat{f}_{P}$.

\begin{assumption}\label{ASS2}~ 
\begin{itemize}
\item[(i)] As $T\rightarrow\infty$  and any$P\ll T$ , the matrix $\frac{\PSI^{\prime}\PSI}{T}$
converges in Frobenius norm to a given positive definite matrix $Q_{P}$,
whose smallest eigenvalue is bounded away from zero. 
\item[(ii)] For every $s\geq0$, $\zeta_{s}\left(P\right)$ exists and $\zeta_{s}\left(P\right)\geq1$
for large enough $P$. 
\end{itemize}
\end{assumption}

The following assumption is needed to obtain consistency of the sieve
estimator.

\begin{assumption} \label{ASS3}~ 
\begin{itemize}
\item[(i)] We have: 
\[
\zeta_{s}\left(P\right)\left(\frac{\left(P\lambda_{nT}\right)\wedge\tau_{nT}}{nT}\right)^{1/2}\rightarrow0,\text{ as }n,T\rightarrow\infty.
\]
\item[(ii)] For $s=0$, we require $N_{P}\rightarrow0$, as $P,T\rightarrow\infty$.
For $s>0$, we require $\zeta_{s}\left(P\right)N_{P}\rightarrow0$,
as $P,T\rightarrow\infty$.
\end{itemize}
\end{assumption}

Assumption \ref{ASS2} (i) restricts the asymptotic behavior of the
matrix of design points $\PSI^{\prime}\PSI$. Notice that this assumption
is not explicitly needed, e.g., in \citet{newey1997}, \citet{jong2002}
and \citet{belloni2014} as they deal with stochastic regressors and
therefore this is guaranteed by an appeal to a Law of Large Numbers.
In particular, Assumption \ref{ASS2} implies that the eigenvalues
of $\frac{\PSI^{\prime}\PSI}{T}$ converge to the eigenvalues of $Q_{P}$,
for a fixed $P$.

\citet{newey1997} and \citet{belloni2014} derive rates of convergence
in probability for $\frac{\PSI^{\prime}\PSI}{T}$ towards the fixed
matrix $Q_{P}$. Newey's result implies that $\mathbb{E}\left\Vert \frac{\PSI^{\prime}\PSI}{T}-Q_{P}\right\Vert _{F}^{2}\le\frac{\zeta_{0}^{2}\left(P\right)P}{T}$.
However, we cannot use directly this result since our regressors are
supposed to be deterministic. However, reasoning as in \citet{reimer1997},
we can see that for any probability measure $\inf_{\left\{ X_{t}\right\} _{t=1}^{T}}\left\Vert \frac{\PSI^{\prime}\PSI}{T}-Q_{P}\right\Vert _{F}^{2}\le\mathbb{E}\left\Vert \frac{\PSI^{\prime}\PSI}{T}-Q_{P}\right\Vert _{F}^{2}$,
so that it is possible to find a point-set $\left\{ X_{t}\right\} _{t=1}^{T}$
such that $\left\Vert \frac{\PSI^{\prime}\PSI}{T}-Q_{P}\right\Vert _{F}\le\zeta_{0}\left(P\right)\sqrt{\frac{P}{T}}$.\footnote{Note that, from \citet{belloni2014} (see their Section 6.1 and Theorem
4.6), one can infer that $\left\Vert \frac{\PSI^{\prime}\PSI}{T}-Q_{P}\right\Vert _{F}\lesssim_{\mathbb{P}}\zeta_{0}\left(P\right)\cdot\sqrt{\frac{P\ln P}{T}}$
from which one gets the weaker bound $\inf_{\left\{ X_{t}\right\} _{t=1}^{T}}\left\Vert \frac{\PSI^{\prime}\PSI}{T}-Q_{P}\right\Vert _{F}^{2}\lesssim\zeta_{0}\left(P\right)\cdot\sqrt{\frac{P\ln P}{T}}$.} Better convergence rates can be obtained in special cases (see the
discussion after Theorem \ref{Th - Main Result}).
\begin{rem}
\label{Rm - Assumption 5} Define the empirical probability $\tilde{\mathbb{P}}^{\left(T\right)}\left(A\right)\triangleq\frac{1}{T}\sum_{t=1}^{T}\mathsf{1}\left\{ X_{t}\in A\right\} $.
A particular case is the one in which the points $\left\{ X_{t}\right\} _{i=1}^{T}\in\mathcal{X}$
are chosen in such a way that their empirical probability converges
to an \emph{asymptotic design measure} $\tilde{\mathbb{P}}$ on $\mathcal{X}$.
In this scenario, we can obtain explicitly $Q_{P}$ as $Q_{P}=\tilde{\mathbb{E}}\left[\psi_{P}^{\prime}(x)\psi_{P}(x)\right]$,
where $\tilde{\mathbb{E}}$ is the expectation taken with respect
to the probability $\tilde{\mathbb{P}}$. This situation is similar
to the one in \citet{cox1988}. In our case, it is unnecessary to
specify the asymptotic design measure but, when available, it can
be used to derive an explicit expression for $Q_{P}$. 
\end{rem}

Assumptions \ref{ASS2} (ii) and \ref{ASS3} (ii) are used to bound
the approximation error and to define a uniform upper bound on the
derivative of the vector of basis functions, as measured through the
Sobolev norm, both of which are standard assumptions in the sieve
literature. When the function $f\left(\cdot\right)$ is taken to be
$c$ times continuously differentiable, we can take $N_{P}=O\left(P^{-c/d}\right)$
\citep[see][]{newey1997,huang2003,chen2007,belloni2014}.

Assumption \ref{ASS3} (i) restricts the behavior of the largest eigenvalue
and the trace of the average covariance matrix of the errors $\bar{\Sigma}$.
Depending on the structure of the covariance matrix, these quantities
may or may not be uniformly bounded away from infinity. 

Under the previous assumptions, it is possible to derive an upper
bound for the convergence rate of $\hat{f}_{P}$ to $f$. 
\begin{thm}
\label{Th - Main Result} Under Assumptions \ref{ASS1} and \ref{ASS2},
the following bound can be established: 
\[
\left|\hat{f}_{P}-f\right|_{s}=O_{\mathbb{P}}\left(\zeta_{s}\left(P\right)\left(\frac{\left(P\lambda_{nT}\right)\wedge\tau_{nT}}{nT}\right)^{1/2}+\zeta_{s}\left(P\right)N_{P}\right).
\]
Under Assumption \ref{ASS3}, the bound converges to 0 and the sieve
estimator is a consistent estimator of $f$ in the $\left|\cdot\right|_{s}$
norm.
\end{thm}

Assumption \ref{ASS3} ensures that the upper bound of Theorem \ref{Th - Main Result}
is $o\left(1\right)$ as $n,T\rightarrow\infty$. Notice that, if
$n$ is taken to be finite and $\lambda_{nT}$ and $\tau_{nT}$ are
uniformly bounded away from infinity, the upper bound of the variance
in Theorem \ref{Th - Main Result} is the same as in \citet{newey1997}. 

\subsection*{Remarks on $\lambda_{nT}$ and $\tau_{nT}$}

The terms $\lambda_{nT}$ and $\tau_{nT}$ that enter the formulas
have a behavior that can be clarified in some cases of interest. 

A first case arises when the answers for a given individual are supposed
to be uncorrelated to each other (even if heteroskedastic), so that
$\Sigma_{i}=\textrm{diag}\left(\sigma_{i1}^{2},\dots,\sigma_{iT}^{2}\right)$.
Define $\sigma_{nT}^{2}=\max_{1\le t\le T}\left(\frac{1}{n}\sum_{i=1}^{n}\sigma_{it}^{2}\right)$.
In this case, $\lambda_{nT}=\sigma_{nT}^{2}$, and $\tau_{nT}\le T\sigma_{nT}^{2}$,
so that the bound in Theorem \ref{Th - Main Result} yields:
\[
\left|\hat{f}_{P}-f\right|_{s}=O_{\mathbb{P}}\left(\zeta_{s}\left(P\right)\left(\frac{P}{nT}\right)^{1/2}\sigma_{nT}+\zeta_{s}\left(P\right)N_{P}\right).
\]
If $\sigma_{nT}^{2}$ is bounded, the bound does not make any difference
between $T$ and $n$. Assumption \ref{ASS2} (i) requires that $T\rightarrow\infty$,
so that also $n=1$ is sufficient to ensure consistency, provided
that $P/T\rightarrow0$. In this case our estimator reduces to the
one in \citet{newey1997}. Despite our assumptions are not comparable
with the ones in \citet{newey1997}, as we consider the case of deterministic
regressors, ours are generally weaker as we require uncorrelatedness
and boundedness of the variances of the errors whereas he requires
independence and identical distribution. Lack of dependence between
tasks and boundedness of the maximum variance allow one to estimate
consistently individual response functions. The hypothesis that the
average covariance matrix of the errors $\bar{\Sigma}$ is diagonal
could be in principle tested. However, we have not been able to find
in the literature a test valid under our conditions, i.e. for possibly
non identically distributed error vectors of increasing dimension.
We leave the development of such a test to future work.

A second case of interest arises when the errors for every individual
have a factor structure, i.e. every error term $\varepsilon_{it}$
can be written as $\varepsilon_{it}=\nu_{i}+\upsilon_{it}$ where
$\mathbb{V}\left(\nu_{i}\right)=\sigma_{\nu}^{2}$, $\mathbb{V}\left(\upsilon_{it}\right)=\sigma_{\upsilon}^{2}$,
$\textrm{Cov}\left(\upsilon_{it},\upsilon_{i^{\prime}t}\right)=0$
for $i\neq i^{\prime}$, and $\textrm{Cov}\left(\nu_{i},\upsilon_{it}\right)=0$,
for all $i$. This means that the matrix $\Sigma_{i}$ can be written
as 
\[
\Sigma_{i}=\sigma_{\nu}^{2}ee^{\prime}+\sigma_{\upsilon}^{2}I_{T},\text{ and }\bar{\Sigma}=\sigma_{\nu}^{2}ee^{\prime}+\sigma_{\upsilon}^{2}I_{T}.
\]
According to Lemma 2.1 in \citet{magnus1982}, $\lambda_{nT}=\sigma_{\nu}^{2}T+\sigma_{\upsilon}^{2}$,
and we obtain $\lambda_{nT}\asymp\tau_{nT}\asymp T$. The bound in
Theorem \ref{Th - Main Result} implies that

\[
\left|\hat{f}_{P}-f\right|_{s}=O_{\mathbb{P}}\left(\zeta_{s}\left(P\right)\cdot\left(n^{-1/2}+N_{P}\right)\right).
\]
If the number of tasks $T$ is held fixed and the number of respondents
is allowed to diverge to infinity, one could strengthen Assumption
\ref{ASS2} (i) to allow for the design matrix to be nonsingular for
any finite value of $P$ and $T$ such that $P\leq T$. However, the
finiteness of $P$ implies that the bias component does not disappear
asymptotically and thus the estimator is not consistent. Similarly,
if the number of individuals is held fixed and $T\rightarrow\infty$,
while the bias vanishes, the variance does not disappear asymptotically,
and again the estimator is not consistent. In our framework, however,
letting both $n,T\rightarrow\infty$ yields a consistent estimator
of $f\left(\cdot\right)$. Parametric rates of convergence can be
achieved for $s=0,$ when $N_{P}=O\left(n^{-1/2}\right)$, i.e. when
both $T$ and $P$ diverge sufficiently fast. Nonparametric rates
for specific choices of the sieve basis are discussed in the Appendix.

\subsection*{}

\begin{example}
{[}Stevens' Model; Example \ref{Ex - Stevens Model - 1} continued{]}\label{Ex - Stevens Model - 2}
We suppose that Stevens' model is the true model that generates the
data. We estimate the model as
\[
\ln p=\kappa\ln\left(X_{1}/X_{2}\right)+\varepsilon.
\]
This is a case of parametric estimation. Assumptions \ref{ASS1} and
\ref{ASS2} are supposed to be true. Moreover, $\zeta_{s}$ is finite
and $N_{P}$ is $0$. The rate of convergence is therefore $\left|\hat{f}_{P}-f\right|_{s}=O_{\mathbb{P}}\left(\left(\frac{\lambda_{nT}}{nT}\right)^{1/2}\right)$.
\end{example}

\subsection*{}

\begin{example}
{[}Cumulative Prospect Theory; Example \ref{Ex - CPT - 1} continued{]}\label{Ex - CPT - 2}
Consider the model of Example \ref{Ex - CPT - 1}. We consider the
estimation of a model of the form:
\[
PE_{it}=f\left(\ln CE_{t},\ln X_{t}\right)+\varepsilon_{it}
\]
using a tensor product of Legendre polynomials:
\[
f_{P}\left(\ln\left(CE_{t}\right),\ln\left(X_{t}\right)\right)=\sum_{j=0}^{J}\sum_{k=0}^{K}\beta_{jk}\cdot L_{j}\left(\ln X_{t}\right)\cdot L_{k}\left(\ln CE_{t}\right).
\]
If the parameters $\beta_{jk0}$ are chosen as in Section \ref{Sect - Model},
we denote the function as $f_{P,0}$. Let us denote as $p_{J\wedge K}$
a polynomial of order $J\wedge K$ and let $p_{J\wedge K,0}$ be the
one with parameters chosen as in Section \ref{Sect - Model}. Then,
$\left\Vert f-f_{P,0}\right\Vert _{\infty}\le\left\Vert f-p_{J\wedge K,0}\right\Vert _{\infty}$.
Using Theorem 2 in \citet{Calvi-Levenberg-JAT-2008} we get $\left\Vert f-p_{J\wedge K,0}\right\Vert _{\infty}=O\left(\sqrt{T}\inf_{p_{J\wedge K}}\left\Vert f-p_{J\wedge K}\right\Vert _{\infty}\right)$.
From \citet{Bos-Levenberg-CMFT-2018} (see also \citet{Trefethen-PAMS-2017}),
supposing that $f$ is analytic, $\inf_{p_{J\wedge K}}\left\Vert f-p_{J\wedge K}\right\Vert _{\infty}=O\left(\rho^{J\wedge K}\right)$
for $\rho<1$. At last, when $s=0$:
\[
N_{P}=\left\Vert f-f_{P,0}\right\Vert _{\infty}=O\left(\sqrt{T}\rho^{J\wedge K}\right).
\]
A similar bound clearly applies when $s>0$. Therefore, using the
fact that $P=\left(J+1\right)\left(K+1\right)\sim KJ$:
\[
\left|\hat{f}_{P}-f\right|_{s}=O_{\mathbb{P}}\left(KJ\left(\left(\frac{\left(KJ\lambda_{nT}\right)\wedge\tau_{nT}}{nT}\right)^{1/2}+T^{1/2}\rho^{J\wedge K}\right)\right).
\]
If the errors have a factor structure, then:
\[
\left|\hat{f}_{P}-f\right|_{s}=O_{\mathbb{P}}\left(JK\left(n^{-1/2}+T^{1/2}\rho^{J\wedge K}\right)\right).
\]
\begin{example}
{[}Stevens' Model; Example \ref{Ex - Stevens Model - 1} continued{]}\label{Ex - Stevens Model - 3}
In this case too, as in Example \ref{Ex - Stevens Model - 2}, we
suppose that Stevens' model holds true but we consider a more general
nonparametric model in log-log form (see \citealp{bernasconi2008}
for a justification):
\[
\ln p=f\left(\ln X_{1},\ln X_{2}\right)+\varepsilon,
\]
where $f$ is an unknown function. We approximate the function $f$
using a polynomial of order $J$ in the two variables $\ln X_{1}$
and $\ln X_{2}$: 
\[
\ln p=\beta_{0}+\beta_{1}\ln X_{1}+\beta_{2}\ln X_{2}+\beta_{3}\ln^{2}X_{1}+\beta_{4}\ln^{2}X_{2}+\beta_{5}\ln X_{1}\ln X_{2}+\cdots+\varepsilon.
\]
This polynomial regression has $P=\frac{\left(J+2\right)\left(J+1\right)}{2}$
parameters. Assumptions \ref{ASS1} and \ref{ASS2} are verified,
provided $\PSI$ is chosen correctly. Provided $P\ge3$, we have $N_{P}=0$
and Assumption \ref{ASS3} (ii) is automatically true. The rate of
convergence is:
\[
\left|\hat{f}_{P}-f\right|_{s}=O_{\mathbb{P}}\left(P^{1+2s}\left(\frac{\left(P\lambda_{nT}\right)\wedge\tau_{nT}}{nT}\right)^{1/2}\right).
\]
If the errors have a factor structure, the bound becomes $\left|\hat{f}_{P}-f\right|_{s}=O_{\mathbb{P}}\left(P^{1+2s}n^{-1/2}\right)$
and Assumption \ref{ASS3} (i) requires $P^{2+4s}n^{-1}\rightarrow0$
to ensure uniform convergence of mixed partial derivatives up to order
$s$, which, for finite $P$, is a standard condition on the sample
size.
\end{example}

\end{example}

\subsection*{}

\subsection*{Regression Models with Individual-Specific Characteristics}

To conclude this section, we briefly discuss the possibility of augmenting
the model to include subject characteristics, which we denote by $Z_{i}$.
In most experiments, these characteristics are inherently discrete
(e.g., age, treatment group, gender, etc.). Assume for simplicity
that the $j-$th element of the vector $Z_{i}$ can take values $\lbrace0,1,\dots,L_{j}\rbrace$
with strictly positive probability. Therefore, every element $Z_{ij}$
of $Z_{i}$ can be decomposed into $L_{j}$ dummy variables, each
one taking value $1$ if $Z_{ij}=l$, and $0$ otherwise, for $l=1,\dots,L_{j}$.
Without loss of generality, we can thus define $Z_{i}\in\lbrace0,1\rbrace^{q}$,
where $q\geq0$ is a positive integer, and $Z_{i}$ can also include
arbitrary interactions between the observed individual characteristics.
For such a binary random vector, we impose that the joint function
$f\left(X_{t},Z_{i}\right)$ is such that $f\left(X_{t},Z_{i}\right)=f_{0}\left(X_{t}\right)$,
whenever all the elements of $Z_{i}$ are equal to $0$. Hence, we
can write any joint function $f\left(X_{t},Z_{i}\right)=f_{0}\left(X_{t}\right)+Z_{i}^{\prime}f_{1}\left(X_{t}\right)$,
which follows from the fact that its value changes in $Z_{ij}$ only
when $Z_{ij}$ is equal to $1$, for $j=1,\dots,q$. This finally
implies the following statistical model

\[
Y_{it}=f_{0}\left(X_{t}\right)+Z_{i}^{\prime}f_{1}\left(X_{t}\right)+\varepsilon_{it},
\]
where the unknown functional coefficients depend on $X_{t}.$ This
nonparametric regression can be cast as a varying coefficient model
\citep[see, e.g.,][]{hastie1993,fan1999,fan2005}. In this flexible
semiparametric framework, we can include an arbitrary number of individual
specific covariates without incurring the curse of dimensionality.
Letting $\tilde{Z_{i}}=\left(1,Z_{i}\right)$, the vector $f=\left(f_{0},f_{1}\right)$
satisfies the following system of moment restrictions

\[
\mathbb{E}\left(\tilde{Z_{i}}Y_{it}\left|X_{t}=x\right.\right)=\mathbb{E}\left(\tilde{Z_{i}}\tilde{Z}_{i}^{\prime}\left|X_{t}=x\right.\right)f\left(x\right)=\mathbb{E}\left(\tilde{Z_{i}}\tilde{Z}_{i}^{\prime}\right)f\left(x\right),
\]
where the last equality follows from the fact that the vector $X_{t}$
should be determined independently of individual's characteristics,
and therefore all the moments of the distribution of $Z_{i}$ are
independent of $X_{t}$. For identification, we only require the additional
condition that the matrix $\mathbb{E}\left(\tilde{Z_{i}}\tilde{Z}_{i}^{\prime}\right)$
is full rank. A nonparametric sieve estimator of the functional coefficients
can be obtained by simply replacing the function $f$ with some finite
dimensional approximation on a space of basis functions, and the unknown
population moments of $Z_{i}$ with their sample counterpart. Estimation
of this model is equivalent to splitting the sample in $2^{q}$ subsamples,
and estimate the unknown regression functions for each one of these
subsamples. However, joint estimation of the vector of coefficients
is naturally more efficient, because it uses the entire sample size.
Hence, the resulting sieve estimators inherit the same properties
as above \citep[see][]{fanzhang2008}.

\section{Asymptotic Normality and Wald Tests}

In the following we investigate the asymptotic normality of functionals
of our nonparametric estimator. These are useful to study the properties
of classical statistical tests. We focus here on the properties of
the Wald test and we provide its bias and the rates of convergence
to its asymptotic distribution: we choose this strategy to be able
to evaluate accurately the interplay between the different asymptotic
parameters appearing below. We do not discuss whether it is possible
to estimate these functionals at $\sqrt{nT}$-rate. Arguably, one
could extend the results in \citet{newey1997} to our setting to provide
such results. 

In the following, we consider estimation of a functional of the function
$f$. As in \citet{andrews1991}, we write this functional as $\Gamma:\mathcal{F}\rightarrow\mathbb{R}^{R}$,
where $R>0$ denotes the number of restrictions. Note that $\Gamma$
is allowed to depend on $n$ and $T$. We provide conditions for asymptotic
normality of the quantity 
\[
W_{nT}=A_{nT}\left(\Gamma\left(\hat{f}_{P}\right)-\Gamma\left(f\right)\right)
\]
where $A_{nT}=V_{nT}^{-1/2}$ and $V_{nT}=\mathbb{V}\left(\Gamma\left(\hat{f}_{P}\right)\right)$.
First of all we consider the case in which $\Gamma$ is linear, then
we will move to the nonlinear case.

\subsection{Linear Case}

We consider both the case in which $R$ is fixed and the case in which
$R$ diverges with $n$ and $T$. The case when $R$ is allowed to
increase with $n$ and $T$ will be particularly useful to derive
the asymptotic properties of Wald tests. We suppose that when applied
to $f_{P}=\PSI\beta$ the functional $\Gamma$ yields: 
\[
\Gamma\left(f_{P}\right)=\Gamma\left(\PSI\beta\right)=\widetilde{\PSI}\beta
\]
where $\widetilde{\PSI}$ can take different values according to the
linear functional $\Gamma$. A full list of examples is in \citet[p. 310]{andrews1991},
but some very simple instances are the following ones: 
\begin{enumerate}
\item \textit{Pointwise evaluation functional}: $\Gamma\left(f\right)=f$,
and $\Gamma\left(f_{P}\right)=\PSI\beta$; 
\item \textit{Pointwise partial derivatives}: $\Gamma\left(f\right)=D^{\lambda}f$
and $\Gamma\left(f_{P}\right)=\left(D^{\lambda}\PSI\right)\beta$; 
\item \textit{Weighted average derivatives}: $\Gamma\left(f\right)=\int_{\mathcal{X}}D^{\lambda}f\left(X\right)\mathrm{d}\nu\left(X\right)$
and $\Gamma\left(f_{P}\right)=\left(\int_{\mathcal{X}}D^{\lambda}\PSI\mathrm{d}\nu\left(X\right)\right)\beta$
for a given probability distribution $\nu$ on $\mathcal{X}$. 
\end{enumerate}
All of these examples can also be considered in vector form, as in
\citet[p. 310]{andrews1991}.

The variance of the linear functional applied to an estimated function,
namely $\Gamma\left(\hat{f}_{P}\right)$, is given by: 
\begin{align*}
V_{nT}= & \mathbb{V}\left[\Gamma\left(\hat{f}_{P}\right)\right]=\frac{1}{n}\widetilde{\PSI}\left(\PSI^{\prime}\PSI\right)^{-1}\PSI^{\prime}\overline{\Sigma}\PSI\left(\PSI^{\prime}\PSI\right)^{-1}\widetilde{\PSI}^{\prime}\\
= & \frac{1}{nT}\widetilde{\PSI}\left(\frac{\PSI^{\prime}\PSI}{T}\right)^{-1}\frac{\PSI^{\prime}\overline{\Sigma}\PSI}{T}\left(\frac{\PSI^{\prime}\PSI}{T}\right)^{-1}\widetilde{\PSI}^{\prime}
\end{align*}
where $\overline{\Sigma}=n^{-1}\sum_{i=1}^{n}\Sigma_{i}$. Provided
$V_{nT}$ is symmetric positive definite, let $A_{nT}=V_{nT}^{-1/2}$
be the symmetric positive definite square root of the inverse of $V_{nT}$.

We decompose $W_{nT}$ in two parts: an error term 
\[
\overline{W}_{nT}\triangleq W_{nT}-\mathbb{E}W_{nT}=A_{nT}\left(\Gamma\left(\hat{f}_{P}\right)-\mathbb{E}\Gamma\left(\hat{f}_{P}\right)\right)=A_{nT}\widetilde{\PSI}\left(\widehat{\beta}-\mathbb{E}\widehat{\beta}\right)
\]
and a bias term 
\[
\mathbb{E}W_{nT}=A_{nT}\left(\widetilde{\PSI}\mathbb{E}\widehat{\beta}-\Gamma\left(f\right)\right).
\]
We provide conditions under which the first term converges in distribution
to a $R-$dimensional standard normal vector and the second term tends
to a null vector. The vector $W_{nT}$ enters in the formula of the
Wald test for the hypothesis $\mathsf{H}_{0}:\Gamma(f)=\Gamma_{0}$:
\begin{align*}
W^{\ast}= & \left(\Gamma\left(\hat{f}_{P}\right)-\Gamma_{0}\right)^{\prime}\left[\mathbb{V}\left(\Gamma\left(\hat{f}_{P}\right)-\Gamma_{0}\right)\right]^{-1}\left(\Gamma\left(\hat{f}_{P}\right)-\Gamma_{0}\right)\\
= & \left(\Gamma\left(\hat{f}_{P}\right)-\Gamma_{0}\right)^{\prime}V_{nT}^{-1}\left(\Gamma\left(\hat{f}_{P}\right)-\Gamma_{0}\right)\\
= & \left(A_{nT}\left(\Gamma\left(\hat{f}_{P}\right)-\Gamma_{0}\right)\right)^{\prime}\left(A_{nT}\left(\Gamma\left(\hat{f}_{P}\right)-\Gamma_{0}\right)\right)\\
= & W_{nT}^{\prime}W_{nT}.
\end{align*}
In the following we will consider a standardized version of the statistic,
namely $W=\frac{W^{\ast}-R}{\sqrt{2R}}$. This is particularly useful
when the number of constraints $R$ is allowed to increase with $n$
and $T$, as we show below.
\begin{example}
{[}Pointwise Constraints{]}\label{Ex - Pointwise Constraints-1} A
first case of interest arises when we want to constrain the function
$f$ at a point. Then $\Gamma$ is a pointwise evaluation functional.
Consider the situation in which $d=1$ and we want to test whether
the function $f$ can be constrained in a point, say $X_{\left(1\right)}$,
to take value $y_{\left(1\right)}$ where $y_{\left(1\right)}=f\left(X_{\left(1\right)}\right)$.
Consider the following regression model using a power series of order
$P$: 
\[
Y_{it}=\sum_{j=0}^{P-1}X_{t}^{j}\beta_{j}+\varepsilon_{it}.
\]
The test concerns the null hypothesis $\mathsf{H}_{0}:y_{\left(1\right)}=f\left(X_{\left(1\right)}\right)$.
For a fixed $P$, we use the Wald test statistic above with 
\[
\underbrace{\widetilde{\PSI}}_{1\times P}=\left[\begin{array}{ccccc}
1 & X_{\left(1\right)} & X_{\left(1\right)}^{2} & \cdots & X_{\left(1\right)}^{P-1}\end{array}\right],\text{ and }\underbrace{\Gamma_{0}}_{1\times1}=y_{(1)}.
\]
\begin{example}
{[}Cumulative Prospect Theory; Example \ref{Ex - CPT - 1} continued{]}\label{Ex - CPT - 3}
As customary in this literature \citep[see, e.g.,][Section 3.1]{luce1999utility},
we assume that $U_{1}\equiv U_{2}$. Moreover, one can suppose that
the utility function can adequately be described by power functions
with exponent $\gamma>0$ \citep[see, e.g.,][Section 3.3]{luce1999utility}.
Therefore:
\begin{equation}
V_{1}\left(\ln CE_{t}\right)-V_{2}\left(\ln X_{t}\right)=\gamma\cdot\left(\ln\left(CE_{t}\right)-\ln\left(X_{t}\right)\right),\label{eq:yaarirs-1-1}
\end{equation}
and the model becomes a semiparametric one with
\[
PE_{it}=F\left(\gamma\left(\ln CE_{t}-\ln X_{t}\right)\right)+\varepsilon_{it}.
\]
We would like to provide a test for the parametric restriction in
equation \eqref{eq:yaarirs-1-1}, for any $\gamma>0$. Under the null
hypothesis, we have that:
\[
f\left(\ln\left(CE_{t}\right),\ln\left(X_{t}\right)\right)=F\left(\gamma\cdot\left(\ln\left(CE_{t}\right)-\ln\left(X_{t}\right)\right)\right).
\]
This null hypothesis cannot be tested directly, as such a test would
require the estimation of the transformation function $F$. Therefore,
we proceed as follows. Under the null, we take the derivatives of
the regression function wrt $\log\left(CE_{t}\right)$ and $\log\left(X_{t}\right)$
respectively. We obtain:
\begin{align*}
\frac{\partial f\left(\ln\left(CE_{t}\right),\ln\left(X_{t}\right)\right)}{\partial\ln\left(CE_{t}\right)}= & \frac{\mathrm{d}F\left(\gamma\cdot\left(\ln\left(CE_{t}\right)-\ln\left(X_{t}\right)\right)\right)}{\mathrm{d}\left(\ln\left(CE_{t}\right)-\ln\left(X_{t}\right)\right)}\cdot\gamma,\\
\frac{\partial f\left(\ln\left(CE_{t}\right),\ln\left(X_{t}\right)\right)}{\partial\ln\left(X_{t}\right)}= & -\frac{\mathrm{d}F\left(\gamma\cdot\left(\ln\left(CE_{t}\right)-\ln\left(X_{t}\right)\right)\right)}{\mathrm{d}\left(\ln\left(CE_{t}\right)-\ln\left(X_{t}\right)\right)}\cdot\gamma.
\end{align*}
Therefore, the sum of these two derivative must be equal to $0$ almost
everywhere. Omitting the arguments of the function $f$ for simplicity,
our null hypothesis can be written as
\[
\frac{\partial f}{\partial\ln\left(X_{t}\right)}+\frac{\partial f}{\partial\ln\left(CE_{t}\right)}=0.
\]
This is a linear functional of the nonparametric estimator. To avoid
the clumsy notation, let us write $x_{1}\equiv\ln\left(X_{t}\right)$
and $x_{2}\equiv\ln\left(CE_{t}\right)$. Let $\psi_{J}\left(x_{1}\right)$
be the row vector of the first $J+1$ Legendre polynomials evaluated
in $x_{1}$. We approximate $f\left(x_{1},x_{2}\right)$ through the
tensor product function $f_{P}\left(x_{1},x_{2}\right)=\left[\psi_{J}\left(x_{1}\right)\otimes\psi_{K}\left(x_{2}\right)\right]\beta$
where $J$ and $K$ can be different. Taking the derivative with respect
to $x_{1}$ we get $\frac{\partial\psi_{J}\left(x_{1}\right)}{\partial x_{1}}=\psi_{J}\left(x_{1}\right)S_{J}$,
where $S_{J}$ is a strictly upper triangular $\left(J+1\right)\times\left(J+1\right)$-matrix
called \emph{operational matrix of differentiation} for Legendre polynomials.
Formulas for this matrix have been derived in \citet{Sparis-Mouroutsos-IJC-1986,Sparis-IEEP-1987},
using a detour through power series, and \citet{Bolek-Auto-1993},
using a direct approach based on formulas for the derivatives of Legendre
polynomials \citep[see also][]{phillips1988}. Therefore:
\begin{align*}
\frac{\partial f_{P}\left(x_{1},x_{2}\right)}{\partial x_{1}} & =\frac{\partial}{\partial x_{1}}\left[\psi_{J}\left(x_{1}\right)\otimes\psi_{K}\left(x_{2}\right)\right]\beta\\
 & =\left[\psi_{J}\left(x_{1}\right)\otimes\psi_{K}\left(x_{2}\right)\right]\left(S_{J}\otimes I_{K+1}\right)\beta.
\end{align*}
From this
\[
\left(\frac{\partial}{\partial x_{1}}+\frac{\partial}{\partial x_{2}}\right)f_{P}\left(x_{1},x_{2}\right)=\left[\psi_{J}\left(x_{1}\right)\otimes\psi_{K}\left(x_{2}\right)\right]\left(S_{J}\otimes I_{K+1}+I_{J+1}\otimes S_{K}\right)\beta
\]
and, at last:
\[
\widetilde{\PSI}=S_{J}\otimes I_{K+1}+I_{J+1}\otimes S_{K}.
\]
Matrices of this kind are sometimes called \emph{Kronecker sums} (see
\citet{Canuto-Simoncini-Verani-LAA-2014,Benzi-Simoncini-SIAMJMAA-2015})
and indicated as $S_{J}\oplus S_{K}$ (the symbol $\oplus$ is sometimes
also used for the direct sum of matrices, as in our Section \ref{Sect - Model}).
Exactly $\left(J+1\right)\wedge\left(K+1\right)$ rows and columns
of the $\left(\left(J+1\right)\left(K+1\right)\right)\times\left(\left(J+1\right)\left(K+1\right)\right)$-matrix
$\widetilde{\PSI}$ are not linearly independent of the others, thus
giving a total of $R=\left(J+1\right)\left(K+1\right)-\left(J+1\right)\wedge\left(K+1\right)=JK+J\vee K$
linearly independent restrictions.
\begin{example}
{[}Stevens' Model; Example \ref{Ex - Stevens Model - 1} continued{]}\label{Ex - Stevens Model - 4}
In order to test Stevens' model, we consider the model of Example
\ref{Ex - Stevens Model - 3}. We want to test the statistical hypothesis
$\mathsf{H}_{0}:f\left(\ln X_{1},\ln X_{2}\right)=\kappa\cdot\ln\left(X_{1}/X_{2}\right)$.
The test is then obtained imposing the constraints $\beta_{0}=\beta_{3}=\cdots=\beta_{P-1}=0$
and $\beta_{1}+\beta_{2}=0$. We use the test statistic above with:
\[
\underbrace{\widetilde{\PSI}}_{(P-1)\times P}=\left[\begin{array}{cccccc}
1 & 0 & 0 & 0 &  & 0\\
0 & 1 & 1 & 0 & \cdots & 0\\
0 & 0 & 0 & 1 &  & 0\\
 & \vdots &  &  & \ddots\\
0 & 0 & 0 & 0 &  & 1
\end{array}\right]\text{ and }\underbrace{\Gamma_{0}}_{(P-1)\times1}=\left[\begin{array}{c}
0\\
\vdots\\
0
\end{array}\right].
\]
In this case $R\asymp P$ and we would like to find conditions such
that $R$ can increase with $n$ and $T$. 
\end{example}

\end{example}

\end{example}

We make the following assumptions.

\begin{assumption} \label{ASS4}~ 
\begin{itemize}
\item[(i)] The function $f\in\mathcal{W}_{c,\infty}$. 
\item[(ii)] $\Gamma$ is a uniformly bounded sequence of linear functionals.
That is, $\Gamma$ is linear and for some constant $0<C_{3}<\infty$
(that can depend on $R$, $n$ or $T$) and integer $s\ge0$ such
that $s\leq c$, one has $\left\Vert \Gamma\left(f\right)\right\Vert _{L_{2}}\le C_{3}\left|f\right|_{s}$
for all $n,T$ and $f\in\mathcal{W}_{s,\infty}$. 
\item When applied to $f_{P}(x)=\psi_{P}(x)\beta$, it yields $\Gamma(f_{P})=\Gamma(\PSI\beta)=\widetilde{\PSI}\beta$. 
\end{itemize}
\end{assumption}

\begin{assumption} \label{ASS5}~ 
\begin{itemize}
\item[(i)] For fixed $n$ and $T$, $\lambda_{\min}\left(\overline{\Sigma}\right)>0$. 
\item[(ii)] The error term $\varepsilon$ is such that 
\[
\frac{T^{1+\delta/2}\max_{i,t}\mathbb{E}\left|\varepsilon_{it}\right|^{2+\delta}}{n^{\delta/2}\lambda_{\min}^{1+\delta/2}\left(\overline{\Sigma}\right)}\rightarrow0,
\]
for some $\delta>0$. 
\end{itemize}
\end{assumption}

\begin{assumption} \label{ASS6} For fixed $n$ and $T$, $\lambda_{\min}\left(\widetilde{\PSI}\widetilde{\PSI}^{\prime}\right)>0$.\end{assumption}
\begin{thm}
\label{Th - Asymptotic Normality}~ Let Assumptions \ref{ASS1}-\ref{ASS6}
hold.
\begin{itemize}
\item[(i)] Let $I_{R}$ be the identity matrix of dimension $R$. Then
\[
\overline{W}_{nT}=A_{nT}\left(\Gamma\left(\hat{f}_{P}\right)-\mathbb{E}\Gamma\left(\hat{f}_{P}\right)\right)\overset{\mathcal{D}}{\longrightarrow}\mathcal{N}\left(0,I_{R}\right).
\]
\item[(ii)] The following bound on the bias can be established:
\[
\left\Vert \mathbb{E}W_{nT}\right\Vert _{L_{2}}\le B_{nT}=\left(\frac{nT}{\lambda_{\min}\left(\overline{\Sigma}\right)}\right)^{1/2}\left[\left\Vert f-f_{P}\right\Vert _{\infty}+\frac{C_{3}\left|f-f_{P}\right|_{s}}{\lambda_{\min}^{1/2}\left(\widetilde{\PSI}\widetilde{\PSI}^{\prime}\right)}\right]
\]
where $s$ is the value of the index for which Assumption \ref{ASS4}
(ii) holds.
\item[(iii)] The standardized Wald test whose test statistic is given by $W=\frac{W_{nT}^{\prime}W_{nT}-R}{\sqrt{2R}}$
can be decomposed as: 
\[
W=\frac{\overline{W}_{nT}^{\prime}\overline{W}_{nT}-R}{\sqrt{2R}}+O_{\mathbb{P}}\left(B_{nT}+\frac{B_{nT}^{2}}{\sqrt{R}}\right)
\]
where the distribution of $\frac{\overline{W}_{nT}^{\prime}\overline{W}_{nT}-R}{\sqrt{2R}}$
is such that 
\[
\left|\mathbb{P}\left\{ \frac{\overline{W}_{nT}^{\prime}\overline{W}_{nT}-R}{\sqrt{2R}}\le\omega\right\} -\mathbb{P}\left\{ \frac{\chi_{R}^{2}-R}{\sqrt{2R}}\le\omega\right\} \right|\le C\frac{R^{1/4}T^{3/2}\max_{i,t}\mathbb{E}\left|\varepsilon_{it}\right|^{3}}{n^{1/2}\lambda_{\min}^{3/2}\left(\overline{\Sigma}\right)}
\]
(irrespectively of the fact that $R$ is fixed or goes to infinity),
or
\[
\left|\mathbb{P}\left\{ \frac{\overline{W}_{nT}^{\prime}\overline{W}_{nT}-R}{\sqrt{2R}}\le\omega\right\} -\mathbb{P}\left\{ Z\le\omega\right\} \right|\le C\left(\frac{R^{1/4}T^{3/2}\max_{i,t}\mathbb{E}\left|\varepsilon_{it}\right|^{3}}{n^{1/2}\lambda_{\min}^{3/2}\left(\overline{\Sigma}\right)}+\frac{1}{R^{1/2}}\right),
\]
where $Z$ is a standard normal random variable. 
\end{itemize}
\end{thm}

\begin{rem}
The bounds in part (iii) of this theorem use the theoretical results
in \citet{bentkus2004}, who provides a Berry-Esséen bound for independent
non-identically distributed random variables. For i.i.d. vectors $\varepsilon_{i}$,
the first part of these bounds would not depend on the number of constraints
$R$ \citep[see][]{bentkus2003}.
\end{rem}

\begin{example}
{[}Pointwise Constraints; Example \ref{Ex - Pointwise Constraints-1}
continued{]}\label{Ex - Pointwise Constraints-2} The linear functional
is $\Gamma(f_{P})=\widetilde{\PSI}\beta$ and $\Gamma(f)=\Gamma_{0}$,
so that $\Gamma^{\prime}=\widetilde{\PSI}$. The variance is given
by 
\[
V_{nT}=\frac{1}{nT}\widetilde{\PSI}\left(\frac{\PSI^{\prime}\PSI}{T}\right)^{-1}\frac{\PSI^{\prime}\overline{\Sigma}\PSI}{T}\left(\frac{\PSI^{\prime}\PSI}{T}\right)^{-1}\widetilde{\PSI}^{\prime}.
\]
Assumptions \ref{ASS1}, \ref{ASS2}, \ref{ASS4} (i) and \ref{ASS6}
hold if $\PSI$ is well-chosen. Letting $s=0$, and supposing that
$N_{P}=\left\Vert f-f_{P}\right\Vert _{\infty}=O\left(P^{-c}\right)$
with $c>0$ and that $\lambda_{nT}\asymp\tau_{nT}\asymp T$, Assumption
\ref{ASS3} holds provided $\frac{P^{2}}{n}\rightarrow0$. Assumption
\ref{ASS4} (ii) holds since $\left\Vert \Gamma(f_{P})\right\Vert _{L_{2}}^{2}=\left[f\left(X_{\left(1\right)}\right)\right]^{2}$
and we can use the upper bound $\left|f\left(X_{\left(1\right)}\right)\right|\le\left\Vert f\right\Vert _{\infty}=\left|f\right|_{0}$
to state that $C_{3}=1$ and $s=0$. Assumption \ref{ASS5} (i) can
be supposed to be true. Then, Assumption \ref{ASS5} (ii) leads to
a constraint on the relative rate of increase of $n$ and $T$: if,
taking $\delta=1$, $\max_{i,t}\mathbb{E}\left|\varepsilon_{it}\right|^{3}$
is bounded from above and $\lambda_{\min}\left(\overline{\Sigma}\right)$
from below uniformly in $n$ and $T$, we need $\frac{T^{3}}{n}\rightarrow0$.
The test statistic is: 
\[
W=\frac{n\frac{\left(\widetilde{\PSI}\widehat{\beta}-\Gamma_{0}\right)^{2}}{\widetilde{\PSI}\left(\PSI^{\prime}\PSI\right)^{-1}\PSI^{\prime}\overline{\Sigma}\PSI\left(\PSI^{\prime}\PSI\right)^{-1}\widetilde{\PSI}^{\prime}}-1}{\sqrt{2}}.
\]
The rate of decrease of the bias is given by: 
\[
B_{nT}=O\left(\left(nT\right)^{1/2}P^{-c}\right)
\]
where we have used $\lambda_{\min}^{1/2}\left(\widetilde{\PSI}\widetilde{\PSI}^{\prime}\right)=\left\Vert \widetilde{\PSI}\right\Vert _{L_{2}}=\left(\sum_{j=0}^{P-1}X_{\left(1\right)}^{2j}\right)^{1/2}\ge1$.
This means that we need $\left(nT\right)^{\frac{1}{2c}}\ll P\ll n^{\frac{1}{2}}$
to get consistency of the sieve estimator and convergence to 0 of
the bias of the Wald test statistic, while asymptotic normality requires
$T^{3}\ll n$.
\begin{example}
{[}Cumulative Prospect Theory; Example \ref{Ex - CPT - 1} continued{]}\label{Ex - CPT - 4}
The matrix $\widetilde{\PSI}$ is the sum of the partial derivatives
of $\PSI$ with respect to $\ln X_{t}$ and $\ln CE_{t}$, respectively.
The variance can thus be written as in Example \ref{Ex - Pointwise Constraints-2}.
Assumptions \ref{ASS1}, \ref{ASS2}, \ref{ASS3}, \ref{ASS4} (i)
and \ref{ASS6} hold whenever $\PSI$ is chosen appropriately. We
assume that $f\in\mathcal{W}_{1,\infty}$, so that $\vert f\vert_{1}<\infty$.
In this example, we further consider the space $\mathcal{H}_{1}\equiv\mathcal{\mathcal{W}}_{1,2}$,
and denote as $\Vert\cdot\Vert_{\mathcal{H}_{1}}$ its Sobolev norm.
Notice that these operators are bounded in $L^{2}$, as a function
in $\mathcal{H}_{1}$ with norm equal to $1$ must admit a derivative
with finite $L^{2}$ norm. Therefore, we obtain

\[
\left\Vert \frac{\partial f}{\partial\ln\left(X\right)}+\frac{\partial f}{\partial\ln\left(CE\right)}\right\Vert _{L^{2}}\leq\left\Vert f\right\Vert _{\mathcal{H}_{1}}\leq C\left|f\right|_{1},
\]
and Assumption \ref{ASS4} (ii) holds with $s=1$. The number of restrictions
$R\sim JK$ and the test statistic can be written as
\[
W=\frac{n\widehat{\beta}^{\prime}\widetilde{\PSI}^{\prime}\left[\widetilde{\PSI}\left(\PSI^{\prime}\PSI\right)^{-1}\PSI^{\prime}\overline{\Sigma}\PSI\left(\PSI^{\prime}\PSI\right)^{-1}\widetilde{\PSI}^{\prime}\right]^{-1}\widetilde{\PSI}\widehat{\beta}-R}{\sqrt{2R}}.
\]
If $f$ is analytic, using Example \ref{Ex - CPT - 2}, the order
of the bias is
\[
B_{nT}=O\left(nTJK\rho^{J\wedge K}\right).
\]
Now, suppose that $B_{nT}=o\left(1\right)$ and that $\max_{i,t}\mathbb{E}\left|\varepsilon_{it}\right|^{3}$
($\lambda_{\min}\left(\overline{\Sigma}\right)$) is bounded uniformly
from above (below). Therefore, $W$ can be approximated by $\frac{\chi_{R}^{2}-R}{\sqrt{2R}}$
if $\frac{JKT^{6}}{n^{2}}\rightarrow0$ and by $\mathcal{N}\left(0,1\right)$
if, in addition, $J,K\rightarrow\infty$.
\begin{example}
{[}Stevens' Model; Example \ref{Ex - Stevens Model - 1} continued{]}\label{Ex - Stevens Model - 5}
The variance can be written as above. Assumptions \ref{ASS1}, \ref{ASS2},
\ref{ASS3}, \ref{ASS4} (i) and \ref{ASS6} hold whenever $\PSI$
is chosen appropriately. Now, $\left\Vert \Gamma\left(f\right)\right\Vert _{L_{2}}^{2}=\beta_{0}^{2}+\left(\beta_{1}+\beta_{2}\right)^{2}+\sum_{j=3}^{\infty}\beta_{j}^{2}$
where the $\beta_{j}$'s are the coefficients in the power series
expansion of the function $f$. Using Young's inequality, we obtain
\[
\left\Vert \Gamma\left(f\right)\right\Vert _{L_{2}}^{2}=\beta_{0}^{2}+\left(\beta_{1}+\beta_{2}\right)^{2}+\sum_{j=3}^{\infty}\beta_{j}^{2}\leq\beta_{0}^{2}+2\beta_{1}^{2}+2\beta_{2}^{2}+\sum_{j=3}^{\infty}\beta_{j}^{2}\leq2\sum_{j=0}^{\infty}\beta_{j}^{2}\leq2\Vert f\Vert_{L_{2}}^{2}\leq2C\Vert f\Vert_{\infty}^{2}.
\]
Therefore, Assumption \ref{ASS4} (ii) holds with $s=0$. Under the
null hypothesis of the test, the bias $B_{nT}$ is exactly $0$. Assumption
\ref{ASS5} (ii) leads to a constraint on the relative rate of increase
of $n$ and $T$: if, for some $\delta>0$, $\max_{i,t}\mathbb{E}\left|\varepsilon_{it}\right|^{2+\delta}$
is bounded from above and $\lambda_{\min}\left(\overline{\Sigma}\right)$
from below uniformly in $n$ and $T$ (thus respecting Assumption
\ref{ASS5} (i)), we need $\frac{T^{2+\delta}}{n^{\delta}}\rightarrow0$.
The test statistic is: 
\[
W=\frac{n\widehat{\beta}^{\prime}\widetilde{\PSI}^{\prime}\left[\widetilde{\PSI}\left(\PSI^{\prime}\PSI\right)^{-1}\PSI^{\prime}\overline{\Sigma}\PSI\left(\PSI^{\prime}\PSI\right)^{-1}\widetilde{\PSI}^{\prime}\right]^{-1}\widetilde{\PSI}\widehat{\beta}-\left(P-1\right)}{\sqrt{2\left(P-1\right)}}.
\]
This can be approximated by $\frac{\chi_{P-1}^{2}-\left(P-1\right)}{\sqrt{2\left(P-1\right)}}$
if $\frac{PT^{6}}{n^{2}}\rightarrow0$ and by $\mathcal{N}\left(0,1\right)$
if, in addition, $P\rightarrow\infty$. 
\end{example}

\end{example}

\end{example}

\subsection{Nonlinear Case}

Now we pass to consider asymptotic normality of nonlinear functionals,
also denoted as $\Gamma:\mathcal{F}\rightarrow\mathbb{R}^{R}$. Our
treatment extends the one in Theorem 2 in \citet{newey1997} to the
multivariate case. We suppose that $\Gamma$ possesses a directional
derivative $\Gamma^{\prime}$ enjoying the properties stated in Assumption
\ref{ASS4} (ii). When the second argument of $\Gamma^{\prime}$ is
the true function $f$, we simply write $\Gamma^{\prime}\left(\cdot,f\right)=\Gamma^{\prime}\left(\cdot\right)$.
In case of nonlinear functionals, we replace Assumption \ref{ASS4}
(ii) with the following. 

\begin{assumption} \label{ASS7} $\Gamma$ is a nonlinear functional
for which a function $\Gamma^{\prime}\left(f;\tilde{f}\right)$ exists
such that: 
\begin{itemize}
\item[(i)] $\Gamma^{\prime}\left(f;\tilde{f}\right)$ is linear in $f$; when
applied to $f_{P}(x)=\psi_{P}(x)\beta$, it yields $\Gamma^{\prime}(f_{P},f)=\Gamma^{\prime}(\PSI\beta)=\widetilde{\PSI}\beta$. 
\item[(ii)] For some $C_{4},C_{5},\epsilon>0$ and for all $\tilde{f}$, $\bar{f}$
such that $\left|f-\tilde{f}\right|_{s}<\epsilon$ and $\left|f-\bar{f}\right|_{s}<\epsilon$,
the inequalities: 
\begin{align*}
\left\Vert \Gamma\left(f\right)-\Gamma\left(\tilde{f}\right)-\Gamma^{\prime}\left(f-\tilde{f};\tilde{f}\right)\right\Vert _{L_{2}}\le & C_{4}\left|f-\tilde{f}\right|_{s}^{2},\\
\left\Vert \Gamma^{\prime}\left(f;\tilde{f}\right)-\Gamma^{\prime}\left(f;\bar{f}\right)\right\Vert _{L_{2}}\le & C_{5}\left|f\right|_{s}\left|\tilde{f}-\bar{f}\right|_{s},
\end{align*}
hold. 
\item[(iii)] $\left\Vert \Gamma^{\prime}\left(g;f\right)\right\Vert _{L_{2}}\le C_{6}\left|g\right|_{s}$
for all $g\in\mathcal{W}_{s,\infty}$ and $s\leq c$. 
\end{itemize}
\end{assumption}

The variance $V_{nT}$ and its square root $A_{nT}$ entering the
statement of the theorem have exactly the same definition as above,
with $\widetilde{\PSI}$ defined as in Assumption \ref{ASS7} (i).
In this case the decomposition of $W_{nT}$ into an error and a bias
term has a parallel in the decomposition into a linear part and a
remainder part (that is not strictly speaking a bias since it is random
and depends upon $\widehat{\beta}$). The linear part provides the
asymptotic distribution and the remainder has to be bounded adequately.
\begin{thm}
\label{Th - Asymptotic Normality under Nonlinearity}~ Let Assumptions
\ref{ASS1}-\ref{ASS3}, \ref{ASS4} (i), \ref{ASS5}-\ref{ASS7}
hold.
\begin{itemize}
\item[(i)] The following asymptotic distribution holds:
\[
\overline{W}_{nT}=A_{nT}\left(\Gamma^{\prime}\left(\hat{f}_{P}\right)-\mathbb{E}\Gamma^{\prime}\left(\hat{f}_{P}\right)\right)\overset{\mathcal{D}}{\longrightarrow}\mathcal{N}\left(0,I_{R}\right).
\]
\item[(ii)] The following bound on the remainder can be established:
\begin{align*}
 & \left\Vert A_{nT}\left(\Gamma\left(\hat{f}_{P}\right)-\Gamma\left(f\right)-\Gamma^{\prime}\left(\hat{f}_{P}\right)+\mathbb{E}\Gamma^{\prime}\left(\hat{f}_{P}\right)\right)\right\Vert _{L_{2}}\\
\le & B_{nT}=C\left(\frac{nT}{\lambda_{\min}\left(\overline{\Sigma}\right)}\right)^{1/2}\left[\left\Vert f-f_{P}\right\Vert _{\infty}+\frac{C_{4}\left|f-\hat{f}_{P}\right|_{s}^{2}+C_{6}\left|f-f_{P}\right|_{s}}{\lambda_{\min}^{1/2}\left(\widetilde{\PSI}\widetilde{\PSI}^{\prime}\right)}\right]
\end{align*}
where $s$ is the value of the index for which Assumption \ref{ASS7}
(iii) holds. 
\item[(iii)] The standardized Wald test whose test statistic is given by $W=\frac{W_{nT}^{\prime}W_{nT}-R}{\sqrt{2R}}$
can be decomposed as: 
\[
W=\frac{\overline{W}_{nT}^{\prime}\overline{W}_{nT}-R}{\sqrt{2R}}+O_{\mathbb{P}}\left(B_{nT}+\frac{B_{nT}^{2}}{\sqrt{R}}\right)
\]
where $\frac{\overline{W}_{nT}^{\prime}\overline{W}_{nT}-R}{\sqrt{2R}}$
respects the Berry-Esséen bounds of Theorem \ref{Th - Asymptotic Normality}
(iii).
\end{itemize}
\end{thm}

\subsection{Estimation of the Variance Matrix}

In this section we provide some results about the effect of replacing
$\overline{\Sigma}$ with an estimator $\widehat{\overline{\Sigma}}$.
Wald-type tests require the replacement to occur in the test statistics
but have standard asymptotic distributions: in this case, we provide
a bound on the absolute error of the replacement in $\overline{W}_{nT}^{\prime}\overline{W}_{nT}$. 
\begin{thm}
\label{Th - Estimation of the Variance Matrix} Define the residuals
\[
\widehat{U}_{i}=Y_{i}-\PSI\widehat{\beta},
\]
and the estimated average covariance matrix $\widehat{\overline{\Sigma}}=\frac{1}{n}\sum_{i=1}^{n}\widehat{U}_{i}\widehat{U}_{i}^{\prime}$.
Consider the setting of Theorem \ref{Th - Asymptotic Normality}.
Let $\widehat{\overline{W}_{nT}^{\prime}\overline{W}_{nT}}$ be the
quantity $\overline{W}_{nT}^{\prime}\overline{W}_{nT}$ obtained replacing
$\overline{\Sigma}$ with $\widehat{\overline{\Sigma}}$. Let Assumptions
\ref{ASS1}-\ref{ASS4} or Assumptions \ref{ASS1}-\ref{ASS3}, \ref{ASS4}(i),
\ref{ASS5}-\ref{ASS7} hold. If
\[
\frac{\lambda_{\max}\left(\widetilde{\PSI}\widetilde{\PSI}^{\prime}\right)}{\lambda_{\min}\left(\widetilde{\PSI}\widetilde{\PSI}^{\prime}\right)}\frac{T\left(1\vee\ln R\right)\sqrt{\max_{i,t}\mathbb{E}\varepsilon_{it}^{4}}}{\sqrt{n}\lambda_{\min}\left(\overline{\Sigma}\right)}\rightarrow0,
\]
then
\[
\frac{\widehat{\overline{W}_{nT}^{\prime}\overline{W}_{nT}}-R}{\sqrt{2R}}=\frac{\overline{W}_{nT}^{\prime}\overline{W}_{nT}-R}{\sqrt{2R}}+O_{\mathbb{P}}\left(\frac{\lambda_{\max}\left(\widetilde{\PSI}\widetilde{\PSI}^{\prime}\right)}{\lambda_{\min}\left(\widetilde{\PSI}\widetilde{\PSI}^{\prime}\right)}\frac{T\sqrt{R}\left(1\vee\ln R\right)\sqrt{\max_{i,t}\mathbb{E}\varepsilon_{it}^{4}}}{\lambda_{\min}\left(\overline{\Sigma}\right)\sqrt{n}}\right).
\]
\end{thm}

\begin{rem}
It is clearly possible to regularize the matrix $\widehat{\overline{\Sigma}}$
before replacing $\overline{\Sigma}$ in $\overline{W}_{nT}^{\prime}\overline{W}_{nT}$.
However, as our asymptotic normality results generally require $T\ll n$
and under this condition it is expected that $\widehat{\overline{\Sigma}}$
has full rank, the regularization should have a limited impact on
the performance of the tests. 
\end{rem}

\begin{example}
{[}Pointwise Constraints; Example \ref{Ex - Pointwise Constraints-1}
continued{]}\label{Ex - Pointwise Constraints-3} Suppose that $\lambda_{\min}\left(\overline{\Sigma}\right)$
is bounded from below and $\max_{i,t}\mathbb{E}\varepsilon_{it}^{4}$
from above uniformly in $n$ and $T$. Since $R=1$, $\lambda_{\max}\left(\widetilde{\PSI}\widetilde{\PSI}^{\prime}\right)=\lambda_{\min}\left(\widetilde{\PSI}\widetilde{\PSI}^{\prime}\right)$
and we have: 
\[
n\frac{\left(\widetilde{\PSI}\widehat{\beta}-\Gamma_{0}\right)^{2}}{\widetilde{\PSI}\left(\PSI^{\prime}\PSI\right)^{-1}\PSI^{\prime}\widehat{\overline{\Sigma}}\PSI\left(\PSI^{\prime}\PSI\right)^{-1}\widetilde{\PSI}^{\prime}}=n\frac{\left(\widetilde{\PSI}\widehat{\beta}-\Gamma_{0}\right)^{2}}{\widetilde{\PSI}\left(\PSI^{\prime}\PSI\right)^{-1}\PSI^{\prime}\overline{\Sigma}\PSI\left(\PSI^{\prime}\PSI\right)^{-1}\widetilde{\PSI}^{\prime}}+O_{\mathbb{P}}\left(\frac{T}{\sqrt{n}}\right).
\]
Since the first term on the right hand side is $O_{\mathbb{P}}\left(1\right)$,
whenever $T/\sqrt{n}\rightarrow0$ the replacement has no asymptotic
effect on the asymptotic distribution of the Wald test. Note that
$T/\sqrt{n}\rightarrow0$ is automatically respected whenever the
condition for asymptotic normality in Example \ref{Ex - Pointwise Constraints-2},
i.e. $\frac{T^{3}}{n}\rightarrow0$, is satisfied.
\begin{example}
{[}Cumulative Prospect Theory; Example \ref{Ex - CPT - 1} continued{]}\label{Ex - CPT - 5}
If $\max_{i,t}\mathbb{E}\left|\varepsilon_{it}\right|^{4}$ ($\lambda_{\min}\left(\overline{\Sigma}\right)$)
is bounded uniformly from above (below), the condition becomes:
\[
\frac{\lambda_{\max}\left(\widetilde{\PSI}\widetilde{\PSI}^{\prime}\right)}{\lambda_{\min}\left(\widetilde{\PSI}\widetilde{\PSI}^{\prime}\right)}\frac{T\ln\left(JK\right)}{\sqrt{n}}\rightarrow0.
\]
Here $\frac{\lambda_{\max}\left(\widetilde{\PSI}\widetilde{\PSI}^{\prime}\right)}{\lambda_{\min}\left(\widetilde{\PSI}\widetilde{\PSI}^{\prime}\right)}$
is the condition number of the Gram matrix $\widetilde{\PSI}$. 

Our derivations contained in the Appendix imply that

\[
\frac{\lambda_{\max}\left(\widetilde{\PSI}\widetilde{\PSI}^{\prime}\right)}{\lambda_{\min}\left(\widetilde{\PSI}\widetilde{\PSI}^{\prime}\right)}\leq J^{4}K^{4}\left(K^{3}+J^{3}\right)72{}^{J+K}.
\]
\begin{example}
{[}Stevens' Model; Example \ref{Ex - Stevens Model - 1} continued{]}\label{Ex - Stevens Model - 6}
Suppose that $\lambda_{\min}\left(\overline{\Sigma}\right)$ is bounded
from below uniformly in $n$ and $T$ and $\max_{i,t}\mathbb{E}\varepsilon_{it}^{4}$
is bounded from above. In this case $\widetilde{\PSI}$ is the matrix
described in Example \ref{Ex - Stevens Model - 4} and $\widetilde{\PSI}\widetilde{\PSI}^{\prime}=\mathrm{diag}\left(1,2,1,\dots1\right)$
so that $\lambda_{\max}\left(\widetilde{\PSI}\widetilde{\PSI}^{\prime}\right)=2$
and $\lambda_{\min}\left(\widetilde{\PSI}\widetilde{\PSI}^{\prime}\right)=1$.
If $\frac{T\ln P}{\sqrt{n}}\rightarrow0$: 
\[
\frac{\widehat{\overline{W}_{nT}^{\prime}\overline{W}_{nT}}-R}{\sqrt{2R}}=\frac{\overline{W}_{nT}^{\prime}\overline{W}_{nT}-R}{\sqrt{2R}}+O_{\mathbb{P}}\left(\frac{T\sqrt{P}\ln P}{\sqrt{n}}\right).
\]
\end{example}

\end{example}

\end{example}

In some cases, it can be of interest to estimate the conditional variance
as a function of observable characteristics of the stimuli and/or
the individuals \citep[see, e.g.,][]{butler2007}. We will investigate
this alternative procedure rather informally. Suppose that the following
model holds:
\[
\mathbb{V}\left(\varepsilon_{it}\right)=\mathbb{E}\varepsilon_{it}^{2}=h\left(\tilde{X}_{it}\right),\quad\quad i=1,\dots,n,\quad t=1,\dots,T,
\]
where the vector $\tilde{X}_{it}$ shares some of the features of
$X_{t}$ and it may also coincides with it. The latter implies, among
other things, that the variance is the same across individuals but
different across questions. 
\begin{example}
{[}Remark \ref{Ex - Heterogeneity in Variances} continued{]} In this
case: 
\begin{align*}
h\left(\tilde{X}_{it}\right)= & \mathbb{V}\left(\varepsilon_{it}\right)=\mathbb{V}\left(f_{i}\left(X_{t}\right)-f\left(X_{t}\right)+\eta_{it}\right)\\
= & \mathbb{E}\left(f\left(X_{t},\omega\right)-f\left(X_{t}\right)\right)^{2}+\mathbb{V}\left(\eta_{it}\right).
\end{align*}
We can write: 
\[
\varepsilon_{it}^{2}=h\left(\tilde{X}_{it}\right)+\nu_{it},\quad\quad i=1,\dots,n,\quad t=1,\dots,T,
\]
where $\nu_{it}=\varepsilon_{it}^{2}-\mathbb{E}\left(\varepsilon_{it}^{2}\right)$.
Unfortunately, the error $\varepsilon_{it}$ is not available but
can be replaced by the residual $\widehat{U}_{it}$. Therefore, we
can use the equation
\end{example}

\[
\widehat{U}_{it}^{2}=h\left(\widetilde{X}_{it}\right)+\nu_{it},\quad\quad i=1,\dots,n,\quad t=1,\dots,T.
\]
If the objective is to estimate the entire structure of the matrix
$\overline{\Sigma}$ as a function $\tilde{X_{it}}$, we are left
with the problem of estimating covariances. A potential solution is
to suppose that errors are equicorrelated, in which case the correlations
can be estimated from the standardized residuals. We do not pursue
this topic here.

\section{Applications\label{Sect - Applications}}

We now apply our estimation procedure to two simple examples introduced
above in Economics and Psychology. In both examples, we estimate the
unknown function nonparametrically, and then we use the Wald statistics
to test meaningful restrictions either on the function itself or on
its derivatives.

The function is approximated using tensor products of Legendre polynomials,
which make the estimation and testing procedures straightforward and
intuitive. We denote as $L_{j}\left(x\right)$, the Legendre polynomial
of order $j$, with $j=0,1,2,\dots$. The order of the polynomial
is chosen by cross-validation \citep{hansen2012}.

\subsection{Cumulative Prospect Theory}

We now use the results in Examples \ref{Ex - CPT - 1}, \ref{Ex - CPT - 2},
\ref{Ex - CPT - 3}, \ref{Ex - CPT - 4} and \ref{Ex - CPT - 5} to
analyze the data of an economic experiment. We employ a classical
experimental design to elicit the preference of an individual in choice
under uncertainty. The elicitation procedure is known as the probability
equivalence method and is dual to the certainty equivalence method.
It works as follows. In a sequence of pairwise comparisons $t$ for
$t=1,\dots,T$, an individual is asked to state the probability $p_{t}$
that would make the individual indifferent between receiving the sure
amount of money $CE_{t}$ or a lottery giving the monetary prize $X_{t}$
with probability $p_{t}$ and $0$ otherwise.

Thus, in the probability equivalence method $CE_{t}$ and $X_{t}$
are the stimuli and $p_{t}$ the response, whereas in the certainty
equivalence method $X_{t}$ and $p_{t}$ are the stimuli and $CE_{t}$
the response. We also remark that, though both methods are in principle
capable to elicit the preferences of the individual, since their early
applications it is known that both methods can lead to various types
of inconsistencies. For this reasons, various proposals of revising
the basic elicitation procedures have been made in the literature
in attempt to control for the inconsistencies (references and discussion
in, e.g., \citealp{hershey1985probability,wakker1996eliciting,abdellaoui2000parameter}).

Both elicitation procedures are nevertheless very simple and are useful
for the purpose of illustrating our nonparametric method of estimation
and inference. We in particular conducted the probability equivalence
experiment with 98 participants ($n=98$). Each of them gave the responses
$p_{t}$'s to 100 questions ($T=100$). The 100 questions employed
monetary prizes $X_{t}$ distributed uniformly between a minimum of
15 Euro and a maximum of 66 Euro, with the sure amount $CE_{t}$ varying
between a minimum of 4 Euro and a maximum of 57 Euro. The experiment
was run individually and was computerized: questions were presented
sequentially to each individual on a computer screen, with the order
of the questions randomized independently for each participant. Each
participant had the opportunity to reconsider the decision to each
individual questions several times before confirming it; but once
confirmed, the computer moved the participant to another question
and previous choices couldn't be any longer revised or accessed. At
the end of each individual experiment one question was randomly selected
for each participant and each participant was paid according to the
choice he or she made in the selected question. In particular, participants
were incentivized to give correct answers by the use of the standard
\citet{becker1964measuring} payment method.

A summary of the distribution of $X_{t}$, $CE_{t}$, and the sample
mean of $p_{t}$ across individuals is reported in the Table \ref{Tab - SumStat CPT}.

\begin{table}[ht]
\centering{}%
\begin{tabular}{lccc}
\hline 
 & $X$ & $CE$ & Av. Prob. \tabularnewline
\hline 
Mean  & 45.81  & 23.82  & 0.63 \tabularnewline
St.Dev  & 14.10  & 14.09  & 0.19 \tabularnewline
Min  & 13.00  & 4.00  & 0.26 \tabularnewline
Max  & 66.00  & 57.00  & 0.92 \tabularnewline
\hline 
\end{tabular}\caption{Summary statistics}
\label{Tab - SumStat CPT}
\end{table}

With these data, we first estimate the following fully nonparametric
regression model:
\[
p_{it}=f\left(\ln\left(CE_{t}\right),\ln\left(X_{t}\right)\right)+\varepsilon_{it},
\]
using a tensor product of polynomials of order $2$ in $\ln\left(CE_{t}\right)$
and polynomials of order $3$ in $\ln\left(X_{t}\right)$. Figure
\ref{fig:yaarifct} depicts the nonparametric estimator of the regression
function $g$.

\begin{figure}[!h] 
\centering 
\includegraphics[scale=0.6,trim = {0 200 0 200}]{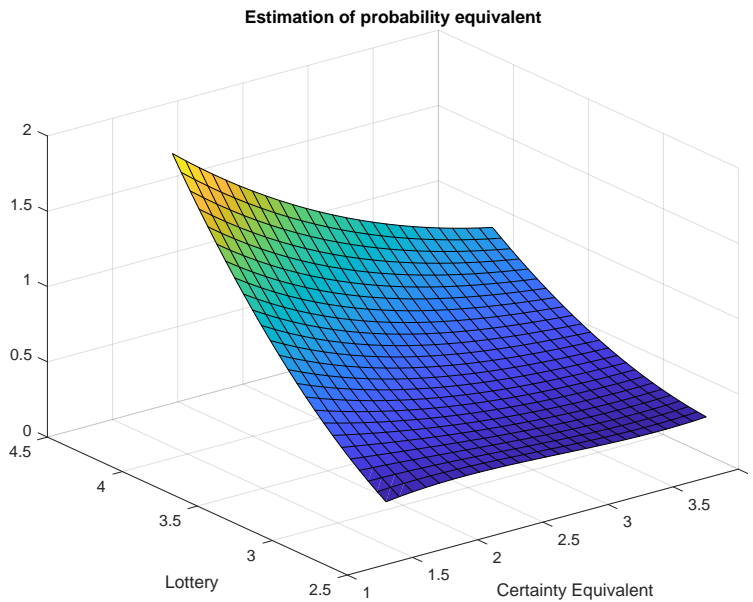} 
\caption{Plot of nonparametric regression for probability equivalent} 
\label{fig:yaarifct} 
\end{figure} 

%Testing the null that the specification is truly additive usually requires to specify the model both under the null and under the alternative, and it is not easy to test in our framework \cite{gozalo2001}. 

To implement the testing procedure, we slightly undersmooth compared
to the estimation above, and take a cubic polynomial in $\ln CE_{t}$
and a quartic polynomial in $\ln X_{t}$ (see Theorem \ref{Th - Asymptotic Normality}).
Therefore, we have a total of $R=18$ restrictions, but only 16 of
them are linearly independent. As $K=3$ and $J=4$, these restrictions
are: 
\[
\beta_{1,3}=\beta_{2,2}=\beta_{2,3}=\beta_{3,1}=\beta_{3,2}=\beta_{3,3}=\beta_{4,0}=\beta_{4,1}=\beta_{4,2}=\beta_{4,3}=0
\]
and:
\[
\beta_{0,1}+\beta_{1,0}=\beta_{0,2}-\beta_{2,0}=\beta_{0,3}+\beta_{3,0}=\beta_{1,2}+\beta_{2,1}=3\beta_{0,2}+\beta_{1,1}=5\beta_{0,3}+\beta_{1,2}=0.
\]
The value of the test statistic is $557.865$, which leads to a rejection
of the null hypothesis, with a $p-$value strictly lower than $0.01$.
This implies that the usual power specification of the utility function
may not be justified in our setting, and such an assumption could
lead to an inconsistent estimation of the probability weighting function.

\subsection{Stevens' Model}

Now we consider an experimental application related to the model in
Examples \ref{Ex - Stevens Model - 1}, \ref{Ex - Stevens Model - 2},
\ref{Ex - Stevens Model - 3}, \ref{Ex - Stevens Model - 4}, \ref{Ex - Stevens Model - 5}
and \ref{Ex - Stevens Model - 6}. The only difference with respect
to the examples is that we use Legendre polynomials, but this has
no impact on the conditions. In the experiment, $n=69$ individuals
were asked to give their estimates of ratios of distances of pairs
of Italian cities from a reference city. Participants were undergraduate
students in economics from the University of Insubria in Italy. We
presented to the subjects 10 pairs of Italian cities and we asked
them to estimate the ratio of their distances with respect to Milan:
the $T=10$ pairs were given by all the possible combinations out
of the five cities Turin, Venice, Rome, Naples and Palermo. The range
of the stimuli goes from 124 to 885 km and the range of the real distance
ratios from 2 to 7.137. We asked participants, first, to state for
each comparison which of the two items they thought was larger, and
then to quantify the relative dominance of the two items, i.e. how
many times the city that they considered more distant from Milan was,
according to them, actually more distant from Milan than the city
they considered less distant. All the experiments were performed in
a random order. The data were already analyzed, with different aims
and techniques, in \citet{Bernasconi-Choirat-Seri-MS-2010,Bernasconi-Choirat-Seri-JMP-2011,Bernasconi-Choirat-Seri-EJOR-2014}.

A summary of the magnitude of the stimuli and of the ratio reported
by the individual is given in Table \ref{sumstat_stevens}.

\begin{table}[ht]
\centering{}%
\begin{tabular}{lccc}
\hline 
 & $\ln X_{1}$  & $\ln X_{2}$  & $\ln p$ \tabularnewline
\hline 
Mean  & 6.45  & 5.47  & 1.20 \tabularnewline
St.Dev  & 0.40  & 0.65  & 0.42 \tabularnewline
Min  & 5.51  & 4.82  & 0.64 \tabularnewline
Max  & 6.79  & 6.49  & 1.85 \tabularnewline
\hline 
\end{tabular}\caption{Summary statistics}
\label{sumstat_stevens} 
\end{table}

A nonparametric series estimator of the function $f$ of Example \ref{Ex - Stevens Model - 3}
is reported in Figure \ref{fig:fct2}. We take quadratic polynomials
in both $\ln X_{1t}$ and $\ln X_{2t}$.

\begin{figure}[!h]
\centering
\includegraphics[scale=0.6,trim = {0 200 0 200}]{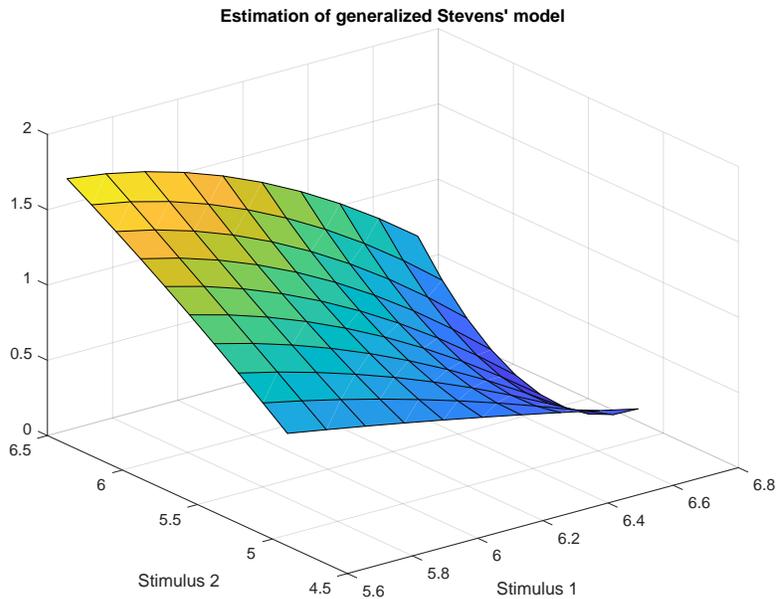}
\caption{Plot of nonparametric regression for generalized Stevens' model}
\label{fig:fct2}
\end{figure}

We test Stevens' model (see Examples \ref{Ex - Stevens Model - 1}
and \ref{Ex - Stevens Model - 2}), in which $f$ is a linear function
of the difference of the log between the two stimuli. That is, our
null hypothesis is:
\[
H_{0}:f\left(\ln X_{1},\ln X_{2}\right)=\kappa\cdot\left(\ln X_{1}-\ln X_{2}\right),
\]
for some real $\kappa$. The test for this hypothesis has been considered
in Examples \ref{Ex - Stevens Model - 4}, \ref{Ex - Stevens Model - 5}
and \ref{Ex - Stevens Model - 6}. Under the null hypothesis, the
intercept of the model is equal to zero; the two slopes sum up to
zero; and all other higher-order coefficients are equal to zero. We
therefore tests $R=8$ restrictions on the vector of estimated parameters.
We remark that the number of individuals $n$ is much larger than
the number of questions $T$, thus providing some support for the
conditions in Example \ref{Ex - Stevens Model - 5}.

In this case, we do not undersmooth to implement our testing procedure,
since the quadratic tensor product basis already exhausts the degrees
of freedom in our model. The value of the test statistic is $1781.218$,
which leads to a rejection of the null hypothesis, with a $p-$value
strictly lower than $0.01$.

\section{Conclusions}

We present in this paper a set of statistical tools useful for the
analysis of some experimental data when the researcher aims at estimating
and testing the average individual response function nonparametrically.
In lack of a theory of errors in either economic or psychophysical
experiments, we allow our regression errors to be correlated within
individual tasks, and to feature heteroskedasticity between different
individuals. In particular, and differently from a large body of literature
on nonparametric regressions, we do not assume that variance of the
error term is uniformly bounded away from infinity. This approach
makes the tools we suggest robust to several structures of the error
covariance matrix, for which theory often provides scarce or contradictory
information. We finally point out that our results can be considered
a worst-case scenario. That is, we only consider the case when the
same tasks are submitted to all individuals. We conjecture that better
asymptotic properties could be obtained if one allows for different
individuals to perform different tasks. We defer the analysis of such
a case to further research.

\newpage{}

\bibliographystyle{agsm}
\bibliography{SMEE}

\newpage{}

\section{Appendix}

\subsection{Rates of convergence for specific choices of the sieve basis in Theorem
\ref{Th - Main Result}}

\subsection*{Parametric estimation}

We briefly consider the case of parametric estimation. In this case,
we have $N_{P}=0$, and $P$ and $\zeta_{s}\left(P\right)$ fixed.
It can be interesting to remark that in order to have consistency
in this context Assumption \ref{ASS3} (ii) is not even necessary.
The convergence rate is: 
\[
\left|\hat{f}_{P}-f\right|_{s}=O_{\mathbb{P}}\left(\left(\frac{\lambda_{nT}}{nT}\right)^{1/2}\right).
\]
It is further possible to show that this is near to the correct rate
of convergence. Indeed, we have: 
\begin{align*}
\mathbb{V}\left(\hat{f}_{P}\left(x\right)-f\left(x\right)\right)= & \mathbb{V}\left(\hat{f}_{P}\left(x\right)-\mathbb{E}\hat{f}_{P}\left(x\right)\right)=\mathbb{V}\left(\psi_{P}\left(x\right)\left(\widehat{\beta}-\beta_{0}\right)\right)\\
= & \frac{1}{\left(nT\right)^{2}}\psi_{P}\left(x\right)\left(\frac{\PSI^{\prime}\PSI}{T}\right)^{-1}\left[\PSI^{\prime}\left(\sum_{i=1}^{n}\Sigma_{i}\right)\PSI\right]\left(\frac{\PSI^{\prime}\PSI}{T}\right)^{-1}\psi_{P}^{\prime}\left(x\right)\\
\ge & \frac{1}{nT}\psi_{P}\left(x\right)\left(\frac{\PSI^{\prime}\PSI}{T}\right)^{-1}\psi_{P}^{\prime}\left(x\right)\lambda_{\min}\left(\overline{\Sigma}\right)
\end{align*}
and, using the assumptions stated above, the only difference is the
replacement of $\lambda_{\min}\left(\overline{\Sigma}\right)$ with
$\lambda_{nT}$.

\subsection*{Fully nonparametric models - Power series\label{subsec:power-series-1}}

Let us look at what happens for power series. In this case, the order
of the polynomial $J$ is linked to $P$ through the relation $P=\frac{\left(d+J\right)!}{d!J!}\asymp J^{d}$,
where $d$ is the number of regressors and the asymptotic equivalence
holds for $J\rightarrow\infty$: this means that it would be possible
to obtain bounds in $J$ from bounds in $P$. In this case \citep[see][p. 157]{newey1997},
$\zeta_{s}\left(P\right)=O\left(P^{1+2s}\right)$, while the two best
known results for $N_{P}$ are $N_{P}=O\left(P^{-c/d}\right)$ for
$s=0$, and $N_{P}=O\left(P^{-\left(c-s\right)}\right)$ for $d=1$,
where $c$ is the number of continuous derivatives of $f$.

When $s=0$, the bound becomes: 
\[
\left|\hat{f}_{P}-f\right|_{0}=O_{\mathbb{P}}\left(P\left(\frac{\left(P\lambda_{nT}\right)\wedge\left(\tau_{nT}\right)}{nT}\right)^{1/2}+P^{1-c/d}\right).
\]
Suppose first that $P=o\left(\frac{\tau_{nT}}{\lambda_{nT}}\right)$.
Then the bound is $O_{\mathbb{P}}\left(P^{\frac{3}{2}}\left(\frac{\lambda_{nT}}{nT}\right)^{1/2}+P^{1-c/d}\right)$,
and in order for this to converge to $0$, we need at least $c>d$
and $P=o\left(\left(\frac{nT}{\lambda_{nT}}\right)^{\frac{1}{3}}\right)$.
The best convergence rate can be obtained when $P^{\star}\asymp\left(\frac{nT}{\lambda_{nT}}\right)^{\frac{d}{d+2c}}$.
In this case the bound is: 
\[
\left|\hat{f}_{P^{\star}}-f\right|_{0}=O_{\mathbb{P}}\left(\left(\frac{\lambda_{nT}}{nT}\right)^{\frac{c-d}{d+2c}}\right).
\]
Suppose now that $\frac{\tau_{nT}}{\lambda_{nT}}=o\left(P\right)$.
Then the bound becomes $O_{\mathbb{P}}\left(P\left(\frac{\tau_{nT}}{nT}\right)^{1/2}+P^{1-c/d}\right)$,
and we need at least $c>d$ and $P=o\left(\left(\frac{nT}{\tau_{nT}}\right)^{1/2}\right)$
to get consistency. The best convergence rate is obtained when $P^{\star}\asymp\left(\frac{nT}{\tau_{nT}}\right)^{\frac{d}{2c}}$,
yielding: 
\[
\left|\hat{f}_{P^{\star}}-f\right|_{0}=O_{\mathbb{P}}\left(\left(\frac{\tau_{nT}}{nT}\right)^{\frac{c-d}{2c}}\right).
\]
As concerns $N_{P}$ for analytic functions, we consider only the
case $d=1$ and $s=0$ (see Example \ref{Ex - CPT - 2} for the case
$d>1$): it is well known that $N_{P}=O\left(\rho^{P}\right)$, where
$\rho$ can be explicitly characterized. Moreover, supposing that
$\frac{\tau_{nT}}{nT}=o\left(1\right)$, the rate of convergence for
the optimal $P^{\star}$ is: 
\[
\left|\hat{f}_{P^{\star}}-f\right|_{0}=\begin{cases}
O_{\mathbb{P}}\left(\left(\frac{\lambda_{nT}}{nT}\right)^{1/2}\sqrt{\ln\left(\frac{nT}{\lambda_{nT}}\right)}\right) & \textrm{when }P=o\left(\frac{\tau_{nT}}{\lambda_{nT}}\right)\\
O_{\mathbb{P}}\left(\left(\frac{\tau_{nT}}{nT}\right)^{1/2}\ln\left(\frac{nT}{\tau_{nT}}\right)\right) & \textrm{when }\frac{\tau_{nT}}{\lambda_{nT}}=o\left(P\right)
\end{cases}
\]
where the first bound holds for $P^{\star}=-\frac{\ln\ln\left(\frac{2\ln\left(1/\rho\right)nT}{\lambda_{nT}}\right)+\ln\left(\frac{2\ln\left(1/\rho\right)nT}{\lambda_{nT}}\right)}{2\ln\rho}\left(1+o\left(1\right)\right)$
and the second for $P^{\star}=-\frac{\ln\left(\frac{nT}{\tau_{nT}}\right)}{2\ln\rho}\left(1+o\left(1\right)\right)$.
These are close to the rates of convergence for the parametric case.

We notice that in the case of power series regression, the quantity
$\left\Vert \left(\frac{\PSI^{\prime}\PSI}{T}\right)-Q_{P}\right\Vert _{F}$
can be bounded in a more efficient way than the one after Assumption
\ref{ASS3}: this also provides some hints about how an experiment
can be designed in order to efficiently approximate $Q_{P}$. Suppose
that we can take the compact space $\mathcal{X}$ as the unit hypercube
$\left[0,1\right]^{d}$ and the asymptotic design measure $\tilde{\mathbb{P}}$
(see Remark \ref{Rm - Assumption 5}) as the uniform probability measure
on $\left[0,1\right]^{d}$. Then we can use the results of \citet{klinger1997}
on the polynomial discrepancy of sequences. In order to do so, we
write as $X_{t,k}$ the $k$-th element of $X_{t}$ and as $\psi_{j,P}\left(X_{t}\right)$
the $j$-th element of $\psi_{P}\left(X_{t}\right)$. Therefore, every
element $\psi_{j,P}\left(X_{t}\right)$ can be written as $\psi_{j,P}\left(X_{t}\right)=X_{t}^{\lambda}\triangleq\prod_{k=1}^{d}\left(X_{t,k}\right)^{\lambda_{k}}$
for a certain multi-index $\lambda=\left(\lambda_{1},\dots,\lambda_{d}\right)$
of non-negative integers. Using the definition of the Frobenius norm,
we have: 
\begin{align*}
\left\Vert \left(\frac{\PSI^{\prime}\PSI}{T}\right)-Q_{P}\right\Vert _{F}^{2}= & \sum_{j=1}^{P}\sum_{j^{\prime}=1}^{P}\left\{ \frac{1}{T}\sum_{t=1}^{T}\psi_{j,P}\left(X_{t}\right)\psi_{j^{\prime},P}\left(X_{t}\right)-\tilde{\mathbb{E}}X_{t}^{\lambda}X_{t}^{\lambda^{\prime}}\right\} ^{2}\\
\le & P^{2}\sup_{\lambda,\lambda^{\prime}\in\mathbb{N}^{K}}\left\{ \frac{1}{T}\sum_{t=1}^{T}X_{t}^{\lambda}X_{t}^{\lambda^{\prime}}-\tilde{\mathbb{E}}X_{t}^{\lambda}X_{t}^{\lambda^{\prime}}\right\} ^{2}\\
= & P^{2}\sup_{\lambda+\lambda^{\prime}\in\mathbb{N}^{K}}\left\{ \frac{1}{T}\sum_{t=1}^{T}X_{t}^{\lambda+\lambda^{\prime}}-\tilde{\mathbb{E}}X_{t}^{\lambda+\lambda^{\prime}}\right\} ^{2}\\
\le & P^{2}\left[D\left(\mathcal{P}_{T}\right)\right]^{2}\le P^{2}4^{d}\left[D^{\star}\left(\mathcal{P}_{T}\right)\right]^{2}
\end{align*}
where $D\left(\mathcal{P}_{T}\right)$ and $D^{\star}\left(\mathcal{P}_{T}\right)$
are the unanchored and the star discrepancies of the sample of points.
Therefore: 
\[
\left\Vert \left(\frac{\PSI^{\prime}\PSI}{T}\right)-Q_{P}\right\Vert _{F}\le P2^{d}D^{\star}\left(\mathcal{P}_{T}\right).
\]
Remark that if $\mathcal{P}_{T}$ is a low-discrepancy sequence, the
discrepancy $D^{\star}\left(\mathcal{P}_{T}\right)$ can be made to
converge to $0$ as fast as $\frac{\left(\log T\right)^{d-1}}{T}$,
so that $\left\Vert \left(\frac{\PSI^{\prime}\PSI}{T}\right)-Q_{P}\right\Vert _{F}=O\left(\frac{P(\log T)^{d-1}}{T}\right)$.
On the other hand, random sampling provides a rate of $\left\Vert \frac{\PSI^{\prime}\PSI}{T}-Q_{P}\right\Vert _{F}=O_{\mathbb{P}}\left(\frac{P^{3/2}}{T^{1/2}}\right)$,
if one considers the Frobenius norm \citep[see][pp. 161-162]{newey1997},
or $\left\Vert \left(\frac{\PSI^{\prime}\PSI}{T}\right)-Q_{P}\right\Vert _{L_{2}}=O_{\mathbb{P}}\left(\frac{P\sqrt{\log T}}{T^{1/2}}\right)$,
where $\Vert\cdot\Vert_{L_{2}}$ denotes the spectral norm \citep[see,][Theorem 4.6]{belloni2014}.
Our bound thus improves over existing results, at least in this particular
case.

\subsection*{Fully nonparametric models - Regression splines}

In the case of regression splines \citep[see][p. 160]{newey1997},
$\zeta_{s}\left(P\right)=O\left(P^{\frac{1}{2}+s}\right)$, while
the only two known results for $N_{P}$ are $N_{P}=O\left(P^{-c/d}\right)$
for $s=0$, and $N_{P}=O\left(P^{-\left(c-s\right)}\right)$ for $d=1$,
where $c$ is the number of continuous derivatives of $f$.

We focus on the case $s=0$. Suppose first that $P=o\left(\frac{\tau_{nT}}{\lambda_{nT}}\right)$.
In this case, the bound becomes: 
\[
\left|\hat{f}_{P}-f\right|_{0}=O_{\mathbb{P}}\left(P\left(\frac{\lambda_{nT}}{nT}\right)^{1/2}+P^{\frac{1}{2}-\frac{c}{d}}\right)
\]
In order for this to converge to $0$, we need at least $2c>d$ and
$P=o\left(\left(\frac{nT}{\lambda_{nT}}\right)^{1/2}\right)$. The
best convergence rate can be obtained when $P^{\star}\asymp\left(\frac{nT}{\lambda_{nT}}\right)^{\frac{d}{d+2c}}$.
In this case the bound is: 
\[
\left|\hat{f}_{P^{\star}}-f\right|_{0}=O_{\mathbb{P}}\left(\left(\frac{\lambda_{nT}}{nT}\right)^{\frac{2c-d}{2d+4c}}\right).
\]
Suppose now that $\frac{\tau_{nT}}{\lambda_{nT}}=o\left(P\right)$.
The bound becomes: 
\[
\left|\hat{f}_{P}-f\right|_{0}=O_{\mathbb{P}}\left(P^{\frac{1}{2}}\left(\frac{\tau_{nT}}{nT}\right)^{\frac{1}{2}}+P^{\frac{1}{2}-\frac{c}{d}}\right).
\]
In this case we need $2c>d$ and $P=o\left(\frac{nT}{\tau_{nT}}\right)$.
The best convergence rate can be obtained when $P^{\star}\asymp\left(\frac{\tau_{nT}}{nT}\right)^{-\frac{d}{2c}}$.
In this case, the bound is: 
\[
\left|\hat{f}_{P^{\star}}-f\right|_{0}=O_{\mathbb{P}}\left(\left(\frac{\tau_{nT}}{nT}\right)^{\frac{2c-d}{4c}}\right).
\]

\subsection{Proof of Example \ref{Ex - CPT - 5}}
\begin{example}
In order to get an upper bound on this number, we first provide an
upper bound on the largest eigenvalue $\lambda_{\max}\left(\widetilde{\PSI}\widetilde{\PSI}^{\prime}\right)$.
Recall from Example \ref{Ex - CPT - 3}that $\widetilde{\PSI}=S_{J}\oplus S_{K}$.
Thus
\begin{align*}
\left\Vert \widetilde{\PSI}\right\Vert _{F}^{2} & =\sum_{k=0}^{K-1}\left\lfloor \frac{K+1-k}{2}\right\rfloor +\sum_{j=0}^{J-1}\left\lfloor \frac{J+1-j}{2}\right\rfloor +\sum_{k=0}^{K-1}\sum_{j=0}^{J-1}\left\{ \left(2k+1\right)^{2}\left\lfloor \frac{K+1-k}{2}\right\rfloor +\left(2j+1\right)^{2}\left\lfloor \frac{J+1-j}{2}\right\rfloor \right\} \\
\le & \sum_{k=0}^{K-1}\left(\frac{K+1-k}{2}\right)+\sum_{j=0}^{J-1}\left(\frac{J+1-j}{2}\right)+\sum_{k=0}^{K-1}\sum_{j=0}^{J-1}\left\{ \left(2k+1\right)^{2}\left(\frac{K+1-k}{2}\right)+\left(2j+1\right)^{2}\left(\frac{J+1-j}{2}\right)\right\} \\
= & \frac{1}{4}\left(K^{2}+3K\right)+\frac{1}{4}\left(J^{2}+3J\right)+\frac{1}{12}JK\left(2K^{3}+12K^{2}+K-3\right)+\frac{1}{12}JK\left(2J^{3}+12J^{2}+J-3\right)\\
= & O\left(JK\left(K^{3}+J^{3}\right)\right).
\end{align*}
Let $\psi_{J}^{\mathrm{Leg}}\left(x_{1}\right)$ ($\psi_{J}^{\mathrm{pow}}\left(x_{1}\right)$)
be the $(J+1)\times1$vector of Legendre polynomials (power series)
evaluated at $x_{1}$, that can be written as $\psi_{J}^{\mathrm{Leg}}\left(x_{1}\right)=\psi_{J}^{\mathrm{pow}}\left(x_{1}\right)A_{J}$.
The tensor product of the Legendre polynomials can be thus written
as 
\[
\psi_{J}^{\mathrm{Leg}}\left(x_{1}\right)\otimes\psi_{K}^{\mathrm{Leg}}\left(x_{2}\right)=\left[\psi_{J}^{\mathrm{pow}}\left(x_{1}\right)\otimes\psi_{K}^{\mathrm{pow}}\left(x_{2}\right)\right]\left(A_{J}\otimes A_{K}\right)
\]
and
\[
g_{P}\left(x_{1},x_{2}\right)=\left[\psi_{J}^{\mathrm{Leg}}\left(x_{1}\right)\otimes\psi_{K}^{\mathrm{Leg}}\left(x_{2}\right)\right]\beta=\left[\psi_{J}^{\mathrm{pow}}\left(x_{1}\right)\otimes\psi_{K}^{\mathrm{pow}}\left(x_{2}\right)\right]\left(A_{J}\otimes A_{K}\right)\beta.
\]
Taking the derivative with respect to $x_{1}$ we get $\frac{\partial\psi_{J}^{\mathrm{pow}}\left(x_{1}\right)}{\partial x_{1}}=\psi_{J}^{\mathrm{pow}}\left(x_{1}\right)B_{J}$,
with 

\[
\underbrace{B_{J}}_{\left(J+1\right)\times\left(J+1\right)}=\begin{bmatrix}\begin{array}{ccccc}
0 & 1 & \dots & \dots & 0\\
0 & 0 & 2 & \dots & \vdots\\
\vdots & 0 & 0 & \ddots & 0\\
\vdots & \dots & \ddots & \ddots & J\\
0 & \dots & \dots & 0 & 0
\end{array}\end{bmatrix},
\]
 a super-diagonal matrix. We thus have

\begin{align*}
\frac{\partial f_{J,K}\left(x_{1},x_{2}\right)}{\partial x_{1}} & =\frac{\partial}{\partial x_{1}}\left[\psi_{J}^{\mathrm{pow}}\left(x_{1}\right)\otimes\psi_{K}^{\mathrm{pow}}\left(x_{2}\right)\right]\left(A_{J}\otimes A_{K}\right)\beta\\
 & =\left[\psi_{J}^{\mathrm{pow}}\left(x_{1}\right)\otimes\psi_{K}^{\mathrm{pow}}\left(x_{2}\right)\right]\left(B_{J}\otimes I_{K}\right)\left(A_{J}\otimes A_{K}\right)\beta\\
 & =\left[\psi_{J}^{\mathrm{Leg}}\left(x_{1}\right)\otimes\psi_{K}^{\mathrm{Leg}}\left(x_{2}\right)\right]\left(A_{J}^{-1}\otimes A_{K}^{-1}\right)\left(B_{J}\otimes I_{K}\right)\left(A_{J}\otimes A_{K}\right)\beta,
\end{align*}
where, to avoid notational cluttering, we let $I_{J}$ and $I_{K}$
be the identity matrices of dimension $J+1$ and $K+1$, respectively.
From this
\[
\left(\frac{\partial}{\partial x_{1}}+\frac{\partial}{\partial x_{2}}\right)f_{J,K}\left(x_{1},x_{2}\right)=\left[\psi_{J}^{\mathrm{Leg}}\left(x_{1}\right)\otimes\psi_{K}^{\mathrm{Leg}}\left(x_{2}\right)\right]\left(A_{J}^{-1}\otimes A_{K}^{-1}\right)\left(B_{J}\otimes I_{K}+I_{J}\otimes B_{K}\right)\left(A_{J}\otimes A_{K}\right)\beta
\]
and:
\[
\widetilde{\PSI}=\left(A_{J}^{-1}\otimes A_{K}^{-1}\right)\left(B_{J}\otimes I_{K}+I_{J}\otimes B_{K}\right)\left(A_{J}\otimes A_{K}\right).
\]
This means that the matrix $S_{J}$ defined in Example \ref{Ex - CPT - 3}
is given by $S_{J}=A_{J}^{-1}B_{J}A_{J}$ (see Section 3 in \citet{Sparis-Mouroutsos-IJC-1986}).
The matrix $\widetilde{\PSI}$ has rank $R=JK+J\vee K$, so that we
have:
\begin{align*}
\lambda_{JK+J\vee K}\left(\widetilde{\PSI}\widetilde{\PSI}^{\prime}\right) & =\lambda_{JK+J\vee K}\left(\left(A_{J}^{-1}\otimes A_{K}^{-1}\right)\left(B_{J}\oplus B_{K}\right)\left(A_{J}A_{J}^{\prime}\otimes A_{K}A_{K}^{\prime}\right)\right.\\
 & \qquad\left.\left(B_{J}\oplus B_{K}\right)^{\prime}\left(A_{J}^{-1}\otimes A_{K}^{-1}\right)^{\prime}\right)\\
 & \ge\lambda_{JK+J\vee K}\left(\left(A_{J}^{-1}\otimes A_{K}^{-1}\right)\left(B_{J}\oplus B_{K}\right)\left(B_{J}\oplus B_{K}\right)^{\prime}\left(A_{J}^{-1}\otimes A_{K}^{-1}\right)^{\prime}\right)\\
 & \qquad\cdot\lambda_{\min}\left(A_{J}A_{J}^{\prime}\otimes A_{K}A_{K}^{\prime}\right)\\
 & \ge\lambda_{JK+J\vee K}\left(\left(B_{J}\oplus B_{K}\right)\left(B_{J}\oplus B_{K}\right)^{\prime}\right)\\
 & \qquad\cdot\lambda_{\min}\left[\left(A_{J}^{-1}\otimes A_{K}^{-1}\right)^{\prime}\left(A_{J}^{-1}\otimes A_{K}^{-1}\right)\right]\lambda_{\min}\left(A_{J}A_{J}^{\prime}\otimes A_{K}A_{K}^{\prime}\right)\\
 & \ge\lambda_{JK+J\vee K}\left(\left(B_{J}\oplus B_{K}\right)\left(B_{J}\oplus B_{K}\right)^{\prime}\right)\cdot\frac{\lambda_{\min}\left(A_{J}A_{J}^{\prime}\right)\lambda_{\min}\left(A_{K}A_{K}^{\prime}\right)}{\lambda_{\max}\left(A_{J}^{\prime}A_{J}\right)\lambda_{\max}\left(A_{K}^{\prime}A_{K}\right)}\\
 & =\frac{\lambda_{JK+J\vee K}\left(\left(B_{J}\oplus B_{K}\right)\left(B_{J}\oplus B_{K}\right)^{\prime}\right)}{\kappa_{2}^{2}\left(A_{J}\right)\kappa_{2}^{2}\left(A_{K}\right)}.
\end{align*}
We note that, by Theorem 2 in \citet{Zielke-LAA-1988}:
\[
\kappa_{2}\left(A_{J}\right)=\sqrt{\frac{\lambda_{\max}\left(A_{J}^{\prime}A_{J}\right)}{\lambda_{\min}\left(A_{J}A_{J}^{\prime}\right)}}\le J\cdot\kappa_{1}\left(A_{J}\right).
\]
Now we characterize the matrices $A_{J}$. As we aim at using the
results of \citet{Farouki-CAGD-1991,Farouki-JCAM-2000}, we introduce
the matrices of the following transformations, namely the basis-transformation
matrix $A_{J}^{\star}$ from Bernstein to Legendre polynomials $\psi_{J}^{\mathrm{Leg}}\left(x_{1}\right)=A_{J}^{\star}\psi_{J}^{\mathrm{Ber}}\left(x_{1}\right)$
and the equivalent matrix $A_{J}^{\star\star}$ from power to Bernstein
polynomials $\psi_{J}^{\mathrm{Ber}}\left(x_{1}\right)=A_{J}^{\star\star}\psi_{J}^{\mathrm{pow}}\left(x_{1}\right)$.
Clearly, $A_{J}=A_{J}^{\star}A_{J}^{\star\star}$. Therefore:
\[
\kappa_{2}\left(A_{J}\right)\le J\cdot\kappa_{1}\left(A_{J}^{\star}\right)\kappa_{1}\left(A_{J}^{\star\star}\right)
\]
where $\kappa_{1}\left(A_{J}^{\star}\right)=2^{J}$ (\citet{Farouki-JCAM-2000})
and $\kappa_{1}\left(A_{J}^{\star\star}\right)=\left(J+1\right)\tbinom{J}{\left\lfloor \frac{2\left(J+1\right)}{3}\right\rfloor }2^{\left\lfloor \frac{2\left(J+1\right)}{3}\right\rfloor }\sim\frac{3^{J+1}\sqrt{J}}{2\sqrt{\pi}}$
(\citet{Farouki-CAGD-1991}). Then:
\begin{align*}
\lambda_{JK+J\vee K}\left(\widetilde{\PSI}\widetilde{\PSI}^{\prime}\right) & \ge\frac{\lambda_{JK+J\vee K}\left(\left(B_{J}\oplus B_{K}\right)\left(B_{J}\oplus B_{K}\right)^{\prime}\right)}{J^{2}K^{2}\cdot\kappa_{1}^{2}\left(A_{J}^{\star}\right)\kappa_{1}^{2}\left(A_{J}^{\star\star}\right)\kappa_{1}^{2}\left(A_{K}^{\star}\right)\kappa_{1}^{2}\left(A_{K}^{\star\star}\right)}\\
 & \asymp\frac{\lambda_{JK+J\vee K}\left(\left(B_{J}\oplus B_{K}\right)\left(B_{J}\oplus B_{K}\right)^{\prime}\right)}{J^{3}K^{3}\cdot36^{J+K}}.
\end{align*}
Now, $B_{J}\oplus B_{K}$ is a strictly upper triangular $\left(\left(J+1\right)\left(K+1\right)\right)\times\left(\left(J+1\right)\left(K+1\right)\right)$-matrix.
The first column and the last row of $B_{J}\oplus B_{K}$ are filled
with zeros and can be removed to get an upper triangular $\left(\left(J+1\right)\left(K+1\right)-1\right)\times\left(\left(J+1\right)\left(K+1\right)-1\right)$-matrix
$\tilde{B}$, where the non-zero singular values of these two matrices
are the same. With standard inequalities, it is hard to get a lower
bound on this eigenvalue that is not negative. However, through some
numerical experiments, we have obtained
\[
\lambda_{JK+J\vee K}\left(\tilde{B}\tilde{B}^{\prime}\right)\ge2^{-(J+K)},
\]

which finally implies our result. 
\end{example}

\subsection{Proof of Theorem \ref{Th - Main Result}}

We start remarking that under Assumption \ref{ASS2}: 
\[
\lim_{T\rightarrow\infty}\frac{\PSI^{\prime}\PSI}{T}=Q_{P}
\]
and wlog the matrix $Q_{P}$ can be taken to be the identity matrix
of dimension $P$, $I_{P}$. Therefore, $\lambda_{\min}\left(\frac{\PSI^{\prime}\PSI}{T}\right)$
converges to $1$. We define the indicator function $1_{T}=\mathsf{1}\left\{ \lambda_{\min}\left(\frac{\PSI^{\prime}\PSI}{T}\right)>1/2\right\} $.
Clearly, $\lim_{T\rightarrow\infty}1_{T}=1$. Moreover, under Assumption
\ref{ASS1}, we write $\widehat{\beta}$ as $\widehat{\beta}=\left(\PSI^{\prime}\PSI\right)^{-1}\PSI^{\prime}\overline{\mathbf{Y}}$,
where $\overline{\mathbf{Y}}=\frac{1}{n}\sum_{i=1}^{n}Y_{i}$, and
$\overline{\mathbf{Y}}=f\left(\mathbf{X}\right)+\overline{\varepsilon}$,
where $\overline{\varepsilon}=\frac{1}{n}\sum_{i=1}^{n}\varepsilon_{i}$.

The following lemma is useful in the proof of the main theorem. 
\begin{lem}
\label{Lm - Convergence of Beta} Under Assumptions \ref{ASS1} and
\ref{ASS2}: 
\[
\mathsf{1}_{T}\left\Vert \widehat{\beta}-\beta_{0}\right\Vert =O_{\mathbb{P}}\left(\left(\frac{\left(P\lambda_{nT}\right)\wedge\left(\tau_{nT}\right)}{nT}\right)^{1/2}\right)+O\left(N_{P}\right).
\]
\end{lem}

\begin{proof} We have: 
\begin{align*}
 & \widehat{\beta}-\beta_{0}\\
= & \left(\PSI^{\prime}\PSI\right)^{-1}\PSI^{\prime}\overline{\mathbf{Y}}-\beta_{0}\\
= & \left(\PSI^{\prime}\PSI\right)^{-1}\PSI^{\prime}\left(\overline{\mathbf{Y}}-f\left(\mathbf{X}\right)\right)+\left(\PSI^{\prime}\PSI\right)^{-1}\PSI^{\prime}\left(f\left(\mathbf{X}\right)-f_{P,0}\left(\mathbf{X}\right)\right)\\
= & \frac{1}{T}\left(\frac{\PSI^{\prime}\PSI}{T}\right)^{-1}\PSI^{\prime}\overline{\varepsilon}+\frac{1}{T}\left(\frac{\PSI^{\prime}\PSI}{T}\right)^{-1}\PSI^{\prime}\left(f\left(\mathbf{X}\right)-f_{P,0}\left(\mathbf{X}\right)\right)
\end{align*}
from which: 
\[
\mathsf{1}_{T}\left\Vert \widehat{\beta}-\beta_{0}\right\Vert \le\mathsf{1}_{T}\frac{1}{T}\left\Vert \left(\frac{\PSI^{\prime}\PSI}{T}\right)^{-1}\PSI^{\prime}\overline{\varepsilon}\right\Vert +\mathsf{1}_{T}\frac{1}{T}\left\Vert \left(\frac{\PSI^{\prime}\PSI}{T}\right)^{-1}\PSI^{\prime}\left(f\left(\mathbf{X}\right)-f_{P,0}\left(\mathbf{X}\right)\right)\right\Vert .
\]
For the first term in the sum, we get two different bounds. First,
we proceed as follows: 
\begin{align*}
\mathbb{E}\left[\mathsf{1}_{T}\left\Vert \left(\frac{\PSI^{\prime}\PSI}{T}\right)^{-1}\PSI^{\prime}\overline{\varepsilon}\right\Vert ^{2}\right]= & \mathsf{1}_{T}\mathbb{E}\left[\overline{\varepsilon}^{\prime}\PSI\left(\frac{\PSI^{\prime}\PSI}{T}\right)^{-1}\left(\frac{\PSI^{\prime}\PSI}{T}\right)^{-1}\PSI^{\prime}\overline{\varepsilon}\right]\\
\le & \mathsf{1}_{T}\mathbb{E}\textrm{tr}\left[\overline{\varepsilon}^{\prime}\PSI\left(\frac{\PSI^{\prime}\PSI}{T}\right)^{-1}\PSI^{\prime}\overline{\varepsilon}\right]\lambda_{\max}\left(\left(\frac{\PSI^{\prime}\PSI}{T}\right)^{-1}\right)\\
= & \mathsf{1}_{T}\textrm{tr}\left[\PSI\left(\frac{\PSI^{\prime}\PSI}{T}\right)^{-1}\PSI^{\prime}\frac{1}{n}\sum_{i=1}^{n}\mathbb{E}\left(\varepsilon_{i}\varepsilon_{i}^{\prime}\right)\right]\frac{1}{\lambda_{\min}\left(\frac{\PSI^{\prime}\PSI}{T}\right)}\\
= & \mathsf{1}_{T}\frac{1}{n}\textrm{tr}\left[\PSI\left(\frac{\PSI^{\prime}\PSI}{T}\right)^{-1}\PSI^{\prime}\overline{\Sigma}\right]\frac{1}{\lambda_{\min}\left(\frac{\PSI^{\prime}\PSI}{T}\right)}.
\end{align*}
For the first version of the bound, we use the following majorization:
\begin{align*}
\mathbb{E}\left[\mathsf{1}_{T}\left\Vert \left(\frac{\PSI^{\prime}\PSI}{T}\right)^{-1}\PSI^{\prime}\overline{\varepsilon}\right\Vert ^{2}\right]\le & \mathsf{1}_{T}\frac{1}{n}\textrm{tr}\left[\PSI\left(\frac{\PSI^{\prime}\PSI}{T}\right)^{-1}\PSI^{\prime}\overline{\Sigma}\right]\frac{1}{\lambda_{\min}\left(\frac{\PSI^{\prime}\PSI}{T}\right)}\\
= & \mathsf{1}_{T}\frac{1}{n}\textrm{tr}\left[\left(\frac{\PSI^{\prime}\PSI}{T}\right)^{-1/2}\left(\PSI^{\prime}\overline{\Sigma}\PSI\right)\left(\frac{\PSI^{\prime}\PSI}{T}\right)^{-1/2}\right]\frac{1}{\lambda_{\min}\left(\frac{\PSI^{\prime}\PSI}{T}\right)}\\
\le & \mathsf{1}_{T}\frac{T}{n}\textrm{tr}\left[\left(\frac{\PSI^{\prime}\PSI}{T}\right)^{-1}\left(\frac{\PSI^{\prime}\PSI}{T}\right)\right]\frac{\lambda_{\max}\left(\overline{\Sigma}\right)}{\lambda_{\min}\left(\frac{\PSI^{\prime}\PSI}{T}\right)}\\
= & \mathsf{1}_{T}\frac{T\textrm{tr}\left(I_{P}\right)\lambda_{nT}}{n\lambda_{\min}\left(\frac{\PSI^{\prime}\PSI}{T}\right)}=O\left(\frac{TP\lambda_{nT}}{n}\right).
\end{align*}
This implies, from Markov's inequality, that:
\[
\mathsf{1}_{T}\frac{1}{T}\left\Vert \left(\frac{\PSI^{\prime}\PSI}{T}\right)^{-1}\PSI^{\prime}\overline{\varepsilon}\right\Vert =O_{\mathbb{P}}\left(\sqrt{\frac{P\lambda_{nT}}{nT}}\right).
\]
For the second version of the bound, we write: 
\begin{align*}
\mathbb{E}\left[\mathsf{1}_{T}\left\Vert \left(\frac{\PSI^{\prime}\PSI}{T}\right)^{-1}\PSI^{\prime}\overline{\varepsilon}\right\Vert ^{2}\right]\le & \mathsf{1}_{T}\frac{1}{n}\textrm{tr}\left[\PSI\left(\frac{\PSI^{\prime}\PSI}{T}\right)^{-1}\PSI^{\prime}\overline{\Sigma}\right]\frac{1}{\lambda_{\min}\left(\frac{\PSI^{\prime}\PSI}{T}\right)}\\
= & \mathsf{1}_{T}\frac{T}{n}\textrm{tr}\left[\overline{\Sigma}^{1/2}\PSI\left(\PSI^{\prime}\PSI\right)^{-1}\PSI^{\prime}\overline{\Sigma}^{1/2}\right]\frac{1}{\lambda_{\min}\left(\frac{\PSI^{\prime}\PSI}{T}\right)}\\
\leq & \frac{T}{n}\textrm{tr}\left(\overline{\Sigma}\right)\frac{\lambda_{\max}\left(\PSI\left(\PSI^{\prime}\PSI\right)^{-1}\PSI^{\prime}\right)}{\lambda_{\min}\left(\frac{\PSI^{\prime}\PSI}{T}\right)}=O\left(\frac{T\tau_{nT}}{n}\right)
\end{align*}
where the last inequality comes from idempotence of $\PSI\left(\PSI^{\prime}\PSI\right)^{-1}\PSI^{\prime}$.
From Markov's inequality, we finally get:
\[
\mathsf{1}_{T}\frac{1}{T}\left\Vert \left(\frac{\PSI^{\prime}\PSI}{T}\right)^{-1}\PSI^{\prime}\overline{\varepsilon}\right\Vert =O_{\mathbb{P}}\left(\sqrt{\frac{\tau_{nT}}{nT}}\right).
\]
Similarly, using the idempotence of $\PSI\left(\PSI^{\prime}\PSI\right)^{-1}\PSI^{\prime}$,
we have that 
\begin{align*}
 & \mathsf{1}_{T}\left\Vert \left(\frac{\PSI^{\prime}\PSI}{T}\right)^{-1}\PSI^{\prime}\left(f\left(\mathbf{X}\right)-f_{P,0}\left(\mathbf{X}\right)\right)\right\Vert ^{2}\\
= & \mathsf{1}_{T}\left[\left(f\left(\mathbf{X}\right)-f_{P,0}\left(\mathbf{X}\right)\right)^{\prime}\PSI\left(\frac{\PSI^{\prime}\PSI}{T}\right)^{-1}\left(\frac{\PSI^{\prime}\PSI}{T}\right)^{-1}\PSI^{\prime}\left(f\left(\mathbf{X}\right)-f_{P,0}\left(\mathbf{X}\right)\right)\right]\\
\le & \mathsf{1}_{T}T\left[\left(f\left(\mathbf{X}\right)-f_{P,0}\left(\mathbf{X}\right)\right)^{\prime}\PSI\left(\PSI^{\prime}\PSI\right)^{-1}\PSI^{\prime}\left(f\left(\mathbf{X}\right)-f_{P,0}\left(\mathbf{X}\right)\right)\right]\lambda_{\max}\left(\left(\frac{\PSI^{\prime}\PSI}{T}\right)^{-1}\right)\\
\le & \mathsf{1}_{T}T\left(f\left(\mathbf{X}\right)-f_{P,0}\left(\mathbf{X}\right)\right)^{\prime}\left(f\left(\mathbf{X}\right)-f_{P,0}\left(\mathbf{X}\right)\right)\frac{\lambda_{\max}\left(\PSI\left(\PSI^{\prime}\PSI\right)^{-1}\PSI^{\prime}\right)}{\lambda_{\min}\left(\frac{\PSI^{\prime}\PSI}{T}\right)}\\
= & \frac{\mathsf{1}_{T}}{\lambda_{\min}\left(\frac{\PSI^{\prime}\PSI}{T}\right)}T\left(f\left(\mathbf{X}\right)-f_{P,0}\left(\mathbf{X}\right)\right)^{\prime}\left(f\left(\mathbf{X}\right)-f_{P,0}\left(\mathbf{X}\right)\right)\\
= & O\left(T^{2}N_{P}^{2}\right).
\end{align*}
The result of the lemma follows. \end{proof}

The bound stated in the theorem comes from the majorization:
\begin{align*}
\mathsf{1}_{T}\left|\hat{f}_{P}-f\right|_{s}= & \mathsf{1}_{T}\left|\psi_{P}\left(x\right)\left(\widehat{\beta}-\beta_{0}\right)\right|_{s}+\left|\psi_{P}\left(x\right)\beta_{0}-f\right|_{s}\\
\le & \zeta_{s}\left(P\right)\mathsf{1}_{T}\left\Vert \widehat{\beta}-\beta_{0}\right\Vert +N_{P}\\
= & O_{\mathbb{P}}\left(\zeta_{s}\left(P\right)\left(\frac{\left(P\lambda_{nT}\right)\wedge\left(\tau_{nT}\right)}{nT}\right)^{1/2}\right)+O\left(\left(\zeta_{s}\left(P\right)+1\right)N_{P}\right).
\end{align*}
The theorem follows from Lemma \ref{Lm - Convergence of Beta} and
Assumption \ref{ASS3}.

\subsection{Asymptotic Normality and Wald Tests}

\emph{Proof of Theorem \ref{Th - Asymptotic Normality}.} \emph{(i)}
First of all we show that $A_{nT}$ is well defined. We have $V_{nT}=\frac{1}{n}\widetilde{\PSI}\left(\PSI^{\prime}\PSI\right)^{-1}\PSI^{\prime}\overline{\Sigma}\PSI\left(\PSI^{\prime}\PSI\right)^{-1}\widetilde{\PSI}^{\prime}$,
and: 
\begin{align}
\lambda_{\min}\left(V_{nT}\right)= & \lambda_{\min}\left(\frac{1}{n}\widetilde{\PSI}\left(\PSI^{\prime}\PSI\right)^{-1}\PSI^{\prime}\overline{\Sigma}\PSI\left(\PSI^{\prime}\PSI\right)^{-1}\widetilde{\PSI}^{\prime}\right)\nonumber \\
\ge & \frac{1}{nT}\lambda_{\min}\left(\widetilde{\PSI}\left(\frac{\PSI^{\prime}\PSI}{T}\right)^{-1}\widetilde{\PSI}^{\prime}\right)\lambda_{\min}\left(\overline{\Sigma}\right)\nonumber \\
\ge & \frac{1}{nT}\lambda_{\min}\left(\widetilde{\PSI}\widetilde{\PSI}^{\prime}\right)\lambda_{\min}\left(\left(\frac{\PSI^{\prime}\PSI}{T}\right)^{-1}\right)\lambda_{\min}\left(\overline{\Sigma}\right)\nonumber \\
\ge & \frac{\lambda_{\min}\left(\widetilde{\PSI}\widetilde{\PSI}^{\prime}\right)\lambda_{\min}\left(\overline{\Sigma}\right)}{nT\lambda_{\max}\left(\frac{\PSI^{\prime}\PSI}{T}\right)}>0\label{Eq - Lower Bound on V_nT}
\end{align}
from Assumptions \ref{ASS5} (i) and \ref{ASS6}. Therefore $A_{nT}=V_{nT}^{-1/2}$
is well defined. From Assumption \ref{ASS4} (ii), we have $\Gamma\left(\hat{f}_{P}\right)=\widetilde{\PSI}\left(\PSI^{\prime}\PSI\right)^{-1}\PSI^{\prime}\overline{\mathbf{Y}}$
and $\mathbb{E}\Gamma\left(\hat{f}_{P}\right)=\widetilde{\PSI}\left(\PSI^{\prime}\PSI\right)^{-1}\PSI^{\prime}f\left(\mathbf{X}\right)$.
Therefore: 
\begin{align*}
\Gamma\left(\hat{f}_{P}\right)-\mathbb{E}\Gamma\left(\hat{f}_{P}\right)= & \widetilde{\PSI}\left(\PSI^{\prime}\PSI\right)^{-1}\PSI^{\prime}\overline{\varepsilon}\\
= & \frac{1}{n}\sum_{i=1}^{n}\widetilde{\PSI}\left(\PSI^{\prime}\PSI\right)^{-1}\PSI^{\prime}\varepsilon_{i}.
\end{align*}
We use the Cramér-Wold device (with $v$ such that $\left\Vert v\right\Vert _{L_{2}}=1$)
applied to $\overline{W}_{nT}$:
\begin{align}
v^{\prime}\overline{W}_{nT}= & v^{\prime}A_{nT}\left(\Gamma\left(\hat{f}_{P}\right)-\mathbb{E}\Gamma\left(\hat{f}_{P}\right)\right)\nonumber \\
= & \frac{1}{n}\sum_{i=1}^{n}v^{\prime}A_{nT}\widetilde{\PSI}\left(\PSI^{\prime}\PSI\right)^{-1}\PSI^{\prime}\varepsilon_{i}\nonumber \\
= & \sum_{i=1}^{n}\frac{1}{n^{1/2}}v^{\prime}\left(\theta_{nT}^{\prime}\theta_{nT}\right)^{-1/2}\theta_{nT}^{\prime}\overline{\Sigma}^{-1/2}\varepsilon_{i}\label{Eq - Formulation of W_nT}
\end{align}
where $\theta_{nT}\triangleq\overline{\Sigma}^{1/2}\PSI\left(\PSI^{\prime}\PSI\right)^{-1}\widetilde{\PSI}^{\prime}$.
We verify Lyapunov condition:
\begin{align*}
 & \sum_{i=1}^{n}\mathbb{E}\left|\frac{1}{n^{1/2}}v^{\prime}\left(\theta_{nT}^{\prime}\theta_{nT}\right)^{-1/2}\theta_{nT}^{\prime}\overline{\Sigma}^{-1/2}\varepsilon_{i}\right|^{2+\delta}\\
\le & \frac{1}{n^{1+\delta/2}}\sum_{i=1}^{n}\mathbb{E}\left\{ \left\Vert v\right\Vert _{L_{2}}^{2+\delta}\left\Vert \left(\theta_{nT}^{\prime}\theta_{nT}\right)^{-1/2}\theta_{nT}^{\prime}\overline{\Sigma}^{-1/2}\varepsilon_{i}\right\Vert _{L_{2}}^{2+\delta}\right\} \\
= & \frac{1}{n^{1+\delta/2}}\sum_{i=1}^{n}\mathbb{E}\left|\lambda_{\max}\left[\left(\theta_{nT}^{\prime}\theta_{nT}\right)^{-1/2}\theta_{nT}^{\prime}\overline{\Sigma}^{-1/2}\varepsilon_{i}\varepsilon_{i}^{\prime}\overline{\Sigma}^{-1/2}\theta_{nT}\left(\theta_{nT}^{\prime}\theta_{nT}\right)^{-1/2}\right]\right|^{1+\delta/2}\\
\le & \frac{1}{n^{1+\delta/2}}\sum_{i=1}^{n}\mathbb{E}\left\{ \left|\lambda_{\max}\left(\left(\theta_{nT}^{\prime}\theta_{nT}\right)^{-1/2}\theta_{nT}^{\prime}\theta_{nT}\left(\theta_{nT}^{\prime}\theta_{nT}\right)^{-1/2}\right)\cdot\lambda_{\max}\left(\overline{\Sigma}^{-1/2}\varepsilon_{i}\varepsilon_{i}^{\prime}\overline{\Sigma}^{-1/2}\right)\right|^{1+\delta/2}\right\} \\
= & \frac{1}{n^{1+\delta/2}}\sum_{i=1}^{n}\mathbb{E}\left|\lambda_{\max}\left(\overline{\Sigma}^{-1/2}\varepsilon_{i}\varepsilon_{i}^{\prime}\overline{\Sigma}^{-1/2}\right)\right|^{1+\delta/2}
\end{align*}
where the first and second steps come from properties of norms, the
third from the Courant-Fisher variational property of eigenvalues,
and the fourth from idempotence. An upper bound on this function can
be obtained as follows:
\begin{align}
 & \frac{\sum_{i=1}^{n}\mathbb{E}\left|\lambda_{\max}\left(\overline{\Sigma}^{-1/2}\varepsilon_{i}\varepsilon_{i}^{\prime}\overline{\Sigma}^{-1/2}\right)\right|^{1+\delta/2}}{n^{1+\delta/2}}\nonumber \\
\le & \frac{\sum_{i=1}^{n}\mathbb{E}\left|\varepsilon_{i}^{\prime}\overline{\Sigma}^{-1}\varepsilon_{i}\right|^{1+\delta/2}}{n^{1+\delta/2}}\le\frac{\sum_{i=1}^{n}\mathbb{E}\left|\varepsilon_{i}^{\prime}\varepsilon_{i}\lambda_{\max}\left(\overline{\Sigma}^{-1}\right)\right|^{1+\delta/2}}{n^{1+\delta/2}}\nonumber \\
= & \frac{\sum_{i=1}^{n}\mathbb{E}\left|\varepsilon_{i}^{\prime}\varepsilon_{i}\right|^{1+\delta/2}}{n^{1+\delta/2}\lambda_{\min}^{1+\delta/2}\left(\overline{\Sigma}\right)}\le\frac{\max_{i}\mathbb{E}\left|\sum_{t=1}^{T}\varepsilon_{it}^{2}\right|^{1+\delta/2}}{n^{\delta/2}\lambda_{\min}^{1+\delta/2}\left(\overline{\Sigma}\right)}\nonumber \\
\le & \frac{\max_{i}c_{1+\delta/2}\sum_{t=1}^{T}\mathbb{E}\left|\varepsilon_{it}^{2}\right|^{1+\delta/2}}{n^{\delta/2}\lambda_{\min}^{1+\delta/2}\left(\overline{\Sigma}\right)}\le\frac{T^{1+\delta/2}\max_{i,t}\mathbb{E}\left|\varepsilon_{it}\right|^{2+\delta}}{n^{\delta/2}\lambda_{\min}^{1+\delta/2}\left(\overline{\Sigma}\right)},\label{Eq - Upper Bound on Lyapunov}
\end{align}
where the last step comes from Loéve's $c_{r}-$inequality (we recall
that this is the inequality $\mathbb{E}\left|\sum_{t=1}^{T}X_{t}\right|^{r}\leq c_{r}\sum_{t=1}^{T}\mathbb{E}\left|X_{t}\right|^{r}$
where $c_{r}=1$ if $0<r\leq1$ and $c_{r}=T^{r-1}$ if $r>1$).

\emph{(ii)} Take $\Delta f\triangleq f-f_{P,0}$. By Assumption \ref{ASS4},
we have:
\begin{align*}
\mathbb{E}\Gamma\left(\hat{f}_{P}\right)= & \widetilde{\PSI}\left(\PSI^{\prime}\PSI\right)^{-1}\PSI^{\prime}f\left(\mathbf{X}\right)\\
= & \widetilde{\PSI}\left(\PSI^{\prime}\PSI\right)^{-1}\PSI^{\prime}\left[f_{P,0}\left(\mathbf{X}\right)+\Delta f\left(\mathbf{X}\right)\right]\\
= & \widetilde{\PSI}\left(\PSI^{\prime}\PSI\right)^{-1}\PSI^{\prime}\left[\PSI\beta_{0}+\Delta f\left(\mathbf{X}\right)\right]\\
= & \widetilde{\PSI}\beta_{0}+\widetilde{\PSI}\left(\PSI^{\prime}\PSI\right)^{-1}\PSI^{\prime}\Delta f\left(\mathbf{X}\right)\\
\Gamma\left(f\right)= & \Gamma\left(f_{P,0}+\Delta f\right)=\Gamma\left(f_{P,0}\right)+\Gamma\left(\Delta f\right)=\widetilde{\PSI}\beta_{0}+\Gamma\left(\Delta f\right)
\end{align*}
and:
\[
A_{nT}\left(\mathbb{E}\Gamma\left(\hat{f}_{P}\right)-\Gamma\left(f\right)\right)=A_{nT}\widetilde{\PSI}\left(\PSI^{\prime}\PSI\right)^{-1}\PSI^{\prime}\Delta f\left(\mathbf{X}\right)-A_{nT}\Gamma\left(\Delta f\right).
\]
Let $v$ be a vector such that $\left\Vert v\right\Vert _{L_{2}}=1$.
Using the Courant-Fisher variational property of eigenvalues, we obtain
the inequality:
\begin{align*}
 & \left|v^{\prime}A_{nT}V_{nT}A_{nT}^{\prime}v\right|\\
= & \left|\frac{1}{n}v^{\prime}A_{nT}\widetilde{\PSI}\left(\PSI^{\prime}\PSI\right)^{-1}\PSI^{\prime}\overline{\Sigma}\PSI\left(\PSI^{\prime}\PSI\right)^{-1}\widetilde{\PSI}^{\prime}A_{nT}^{\prime}v\right|\\
\ge & \left|\frac{1}{n}v^{\prime}A_{nT}\widetilde{\PSI}\left(\PSI^{\prime}\PSI\right)^{-1}\PSI^{\prime}\PSI\left(\PSI^{\prime}\PSI\right)^{-1}\widetilde{\PSI}^{\prime}A_{nT}^{\prime}v\right|\\
 & \qquad\cdot\lambda_{\min}\left(\overline{\Sigma}\right)\\
= & \left|\frac{1}{n}v^{\prime}A_{nT}\widetilde{\PSI}\left(\PSI^{\prime}\PSI\right)^{-1}\widetilde{\PSI}^{\prime}A_{nT}^{\prime}v\right|\cdot\lambda_{\min}\left(\overline{\Sigma}\right)
\end{align*}
and from this, using $A_{nT}V_{nT}A_{nT}^{\prime}=I_{R}$:
\begin{equation}
\left|\frac{1}{n}v^{\prime}A_{nT}\widetilde{\PSI}\left(\PSI^{\prime}\PSI\right)^{-1}\widetilde{\PSI}^{\prime}A_{nT}^{\prime}v\right|\le\frac{\left|v^{\prime}A_{nT}V_{nT}A_{nT}^{\prime}v\right|}{\lambda_{\min}\left(\overline{\Sigma}\right)}\le\frac{1}{\lambda_{\min}\left(\overline{\Sigma}\right)}.\label{Eq - Intermediate Inequality}
\end{equation}
We have:
\begin{align*}
 & \left|v^{\prime}A_{nT}\widetilde{\PSI}\left(\PSI^{\prime}\PSI\right)^{-1}\PSI^{\prime}\Delta f\left(\mathbf{X}\right)\right|\\
\le & n^{1/2}\left|\frac{1}{n}v^{\prime}A_{nT}\widetilde{\PSI}\left(\PSI^{\prime}\PSI\right)^{-1}\widetilde{\PSI}^{\prime}A_{nT}^{\prime}v\right|^{1/2}\left|\Delta f\left(\mathbf{X}\right)^{\prime}\PSI\left(\PSI^{\prime}\PSI\right)^{-1}\PSI^{\prime}\Delta f\left(\mathbf{X}\right)\right|^{1/2}\\
\le & \frac{n^{1/2}}{\lambda_{\min}^{1/2}\left(\overline{\Sigma}\right)}\cdot\left|\Delta f\left(\mathbf{X}\right)^{\prime}\Delta f\left(\mathbf{X}\right)\right|^{1/2}\lambda_{\max}^{1/2}\left(\PSI\left(\PSI^{\prime}\PSI\right)^{-1}\PSI^{\prime}\right)\\
\le & \frac{n^{1/2}}{\lambda_{\min}^{1/2}\left(\overline{\Sigma}\right)}\cdot\left|\Delta f\left(\mathbf{X}\right)^{\prime}\Delta f\left(\mathbf{X}\right)\right|^{1/2}\le\frac{\left(nT\right)^{1/2}\left\Vert \Delta f\right\Vert _{\infty}}{\lambda_{\min}^{1/2}\left(\overline{\Sigma}\right)}
\end{align*}
where the first inequality is Cauchy-Schwarz' and the second inequality
uses (\ref{Eq - Intermediate Inequality}) and the idempotence of
$\PSI\left(\PSI^{\prime}\PSI\right)^{-1}\PSI^{\prime}$. Similarly,
\begin{align*}
\left|v^{\prime}A_{nT}\Gamma\left(\Delta f\right)\right|\le & \left|v^{\prime}A_{nT}A_{nT}^{\prime}v\right|^{1/2}\left|\Gamma\left(\Delta f\right)^{\prime}\Gamma\left(\Delta f\right)\right|^{1/2}\\
\le & \left\Vert \Gamma\left(\Delta f\right)\right\Vert _{L_{2}}\lambda_{\max}^{1/2}\left(V_{nT}^{-1}\right)\le\frac{\left\Vert \Gamma\left(\Delta f\right)\right\Vert _{L_{2}}}{\lambda_{\min}^{1/2}\left(V_{nT}\right)}\\
\le & C_{3}\frac{\left|\Delta f\right|_{s}}{\lambda_{\min}^{1/2}\left(V_{nT}\right)}\le\frac{C_{3}\left|\Delta f\right|_{s}\left(nT\right)^{1/2}\lambda_{\max}^{1/2}\left(\frac{\PSI^{\prime}\PSI}{T}\right)}{\lambda_{\min}^{1/2}\left(\widetilde{\PSI}\widetilde{\PSI}^{\prime}\right)\lambda_{\min}^{1/2}\left(\overline{\Sigma}\right)}
\end{align*}
where the first inequality is Cauchy-Schwarz', the second inequality,
i.e. $\left|v^{\prime}A_{nT}A_{nT}^{\prime}v\right|\le\lambda_{\max}\left(V_{nT}^{-1}\right)$,
comes from the Courant-Fisher variational property of eigenvalues,
and the last inequality from (\ref{Eq - Lower Bound on V_nT}). 

\emph{(iii)} We can write: 
\[
W_{nT}^{\prime}W_{nT}=\overline{W}_{nT}^{\prime}\overline{W}_{nT}+2\overline{W}_{nT}^{\prime}\mathbb{E}W_{nT}+\mathbb{E}W_{nT}^{\prime}\mathbb{E}W_{nT},
\]
where 
\[
\left|2\overline{W}_{nT}^{\prime}\mathbb{E}W_{nT}+\mathbb{E}W_{nT}^{\prime}\mathbb{E}w_{nT}\right|\le2\left\Vert \overline{W}_{nT}\right\Vert _{L_{2}}\left\Vert \mathbb{E}W_{nT}\right\Vert _{L_{2}}+\left\Vert \mathbb{E}W_{nT}\right\Vert _{L_{2}}^{2}.
\]
First of all, we want to find conditions under which $\overline{W}_{nT}$
approaches a Gaussian random vector when $n$ and $T$ (and sometimes
$R$) diverge to infinity. We take
\[
\mathbb{P}\left\{ \overline{W}_{nT}^{\prime}\overline{W}_{nT}\le\omega\right\} =\mathbb{P}\left\{ \overline{W}_{nT}\in\mathsf{B}_{\sqrt{\omega}}\right\} ,
\]
where $\mathsf{B}_{\sqrt{\omega}}$ is the ball of radius $\sqrt{\omega}$
in $\mathbb{R}^{R}$. For a standard Gaussian $\left(R\times1\right)-$vector
$g_{R}$ we have 
\[
\Delta_{nT}=\sup_{\omega}\left|\mathbb{P}\left\{ \overline{W}_{nT}^{\prime}\overline{W}_{nT}\le\omega\right\} -\mathbb{P}\left\{ g_{R}^{\prime}g_{R}\le\omega\right\} \right|=\sup_{\omega}\left|\mathbb{P}\left\{ \overline{W}_{nT}\in\mathsf{B}_{\sqrt{\omega}}\right\} -\mathbb{P}\left\{ g_{r}\in\mathsf{B}_{\sqrt{\omega}}\right\} \right|.
\]
To bound this quantity, we express $\overline{W}_{nT}$ as in (\ref{Eq - Formulation of W_nT})
and we use the Berry-Esséen bound of Theorem 1.1 in \citet{bentkus2004}:
\[
\Delta_{nT}\le CR^{1/4}\frac{\sum_{i=1}^{n}\mathbb{E}\left\Vert \left(\theta_{nT}^{\prime}\theta_{nT}\right)^{-1/2}\theta_{nT}^{\prime}\overline{\Sigma}^{-1/2}\varepsilon_{i}\right\Vert _{L_{2}}^{3}}{n^{3/2}}.
\]
Notice that, since $\mathbb{V}\left(W_{nT}\right)=I_{R}$, the normalization
condition of the theorem is respected. The right-hand side of the
previous equation can be majorized using (\eqref{Eq - Upper Bound on Lyapunov})
with $\delta=1$. We get

\begin{equation}
\Delta_{nT}\le CR^{\frac{1}{4}}\frac{T^{3/2}\max_{i,t}\mathbb{E}\left|\varepsilon_{it}\right|^{3}}{n^{1/2}\lambda_{\min}^{3/2}\left(\overline{\Sigma}\right)}.\label{Eq - Berry-Esseen Bound}
\end{equation}
Then, respectively from part \emph{(ii)} and part \emph{(i)} of the
present theorem: 
\begin{align*}
\left\Vert \mathbb{E}W_{nT}\right\Vert _{L_{2}}^{2}\le & B_{nT}^{2}\\
\left\Vert \overline{W}_{nT}\right\Vert _{L_{2}}^{2}= & O_{\mathbb{P}}\left(R\right).
\end{align*}
Summing up: 
\begin{align*}
W=\frac{W_{nT}^{\prime}W_{nT}-R}{\sqrt{2R}}= & \frac{\overline{W}_{nT}^{\prime}\overline{W}_{nT}-R+2\overline{W}_{nT}^{\prime}\mathbb{E}W_{nT}+\mathbb{E}W_{nT}^{\prime}\mathbb{E}W_{nT}}{\sqrt{2R}}\\
= & \frac{\overline{W}_{nT}^{\prime}\overline{W}_{nT}-R}{\sqrt{2R}}+O\left(\frac{\left\Vert \overline{W}_{nT}\right\Vert _{L_{2}}\left\Vert \mathbb{E}W_{nT}\right\Vert _{L_{2}}+\left\Vert \mathbb{E}W_{nT}\right\Vert _{L_{2}}^{2}}{\sqrt{R}}\right)\\
= & \frac{\overline{W}_{nT}^{\prime}\overline{W}_{nT}-R}{\sqrt{2R}}+O_{\mathbb{P}}\left(B_{nT}+\frac{B_{nT}^{2}}{\sqrt{R}}\right).
\end{align*}
The approximation of $\frac{\chi_{M}^{2}-R}{\sqrt{2R}}$ through a
standard normal random variable proceeds using the classical Berry-Esséen
bound for sums of independent identically distributed random variables.

\emph{Proof of Theorem \ref{Th - Asymptotic Normality under Nonlinearity}.}
First of all, we decompose the quantity $\Gamma\left(\hat{f}_{P}\right)-\Gamma\left(f\right)$
as follows: 
\begin{align*}
\Gamma\left(\hat{f}_{P}\right)-\Gamma\left(f\right)= & \Gamma\left(\hat{f}_{P}\right)-\Gamma\left(f\right)-\Gamma^{\prime}\left(\hat{f}_{P}\right)+\Gamma^{\prime}\left(f\right)+\Gamma^{\prime}\left(\hat{f}_{P}\right)-\mathbb{E}\Gamma^{\prime}\left(\hat{f}_{P}\right)+\mathbb{E}\Gamma^{\prime}\left(\hat{f}_{P}\right)-\Gamma^{\prime}\left(f\right)\\
= & \Gamma\left(\hat{f}_{P}\right)-\Gamma\left(f\right)-\Gamma^{\prime}\left(\hat{f}_{P}\right)+\Gamma^{\prime}\left(f\right)+\widetilde{\PSI}\left(\PSI^{\prime}\PSI\right)^{-1}\PSI^{\prime}\left(\overline{\mathbf{Y}}-f\left(\mathbf{X}\right)\right)\\
 & +\mathbb{E}\Gamma^{\prime}\left(\hat{f}_{P}\right)-\Gamma^{\prime}\left(f\right)\\
= & \Gamma\left(\hat{f}_{P}\right)-\Gamma\left(f\right)-\Gamma^{\prime}\left(\hat{f}_{P}\right)+\Gamma^{\prime}\left(f\right)+\widetilde{\PSI}\left(\PSI^{\prime}\PSI\right)^{-1}\PSI^{\prime}\overline{\varepsilon}+\mathbb{E}\Gamma^{\prime}\left(\hat{f}_{P}\right)-\Gamma^{\prime}\left(f\right).
\end{align*}
We will show in the following that the dominating term is $\widetilde{\PSI}\left(\PSI^{\prime}\PSI\right)^{-1}\PSI^{\prime}\overline{\varepsilon}$
while all the others converge to $0$. Therefore, we compute the variance
of the term, that we call 
\begin{align*}
V_{nT}= & \frac{1}{nT}\widetilde{\PSI}\left(\frac{\PSI^{\prime}\PSI}{T}\right)^{-1}\frac{\PSI^{\prime}\overline{\Sigma}\PSI}{T}\left(\frac{\PSI^{\prime}\PSI}{T}\right)^{-1}\widetilde{\PSI}^{\prime}
\end{align*}
where $\overline{\Sigma}=n^{-1}\sum_{i=1}^{n}\Sigma_{i}$. Provided
$V_{nT}$ is symmetric positive definite, let $A_{nT}=V_{nT}^{-1/2}$
be the symmetric positive definite square root of the inverse of $V_{nT}$.
The positive definiteness of $V_{nT}$ is checked as in the proof
of Theorem \ref{Th - Asymptotic Normality} under Assumptions \ref{ASS5}
(i) and \ref{ASS6}. At last, we can decompose the normalized vector
in a linear part (yielding the asymptotic distributional result):
\[
A_{nT}\left(\Gamma^{\prime}\left(\hat{f}_{P}\right)-\mathbb{E}\Gamma^{\prime}\left(\hat{f}_{P}\right)\right)=A_{nT}\widetilde{\PSI}\left(\PSI^{\prime}\PSI\right)^{-1}\PSI^{\prime}\overline{\varepsilon}
\]
and a nonlinear part (that does not contribute to the asymptotic distribution):
\begin{align*}
A_{nT}\left(\mathbb{E}\Gamma\left(\hat{f}_{P}\right)-\Gamma\left(f\right)\right)= & A_{nT}\left\{ \Gamma\left(\hat{f}_{P}\right)-\Gamma\left(f\right)-\Gamma^{\prime}\left(\hat{f}_{P}\right)+\Gamma^{\prime}\left(f\right)\right\} \\
 & +A_{nT}\left\{ \mathbb{E}\Gamma^{\prime}\left(\hat{f}_{P}\right)-\Gamma^{\prime}\left(f\right)\right\} .
\end{align*}
As concerns the linear part, asymptotic normality of $A_{nT}\widetilde{\PSI}\left(\PSI^{\prime}\PSI\right)^{-1}\PSI^{\prime}\overline{\varepsilon}$
is verified as in Theorem \ref{Th - Asymptotic Normality}.

Now we provide an upper bound for the nonlinear part. We start remarking
that:
\begin{align*}
\left\Vert A_{nT}\right\Vert _{L_{2}}^{2}= & \lambda_{\max}\left(V_{nT}^{-1}\right)=\frac{1}{\lambda_{\min}\left(V_{nT}\right)}\\
= & \frac{1}{\lambda_{\min}\left(\frac{1}{n}\widetilde{\PSI}\left(\PSI^{\prime}\PSI\right)^{-1}\PSI^{\prime}\overline{\Sigma}\PSI\left(\PSI^{\prime}\PSI\right)^{-1}\widetilde{\PSI}^{\prime}\right)}\\
\le & \frac{n}{\lambda_{\min}\left(\widetilde{\PSI}\left(\PSI^{\prime}\PSI\right)^{-1}\widetilde{\PSI}^{\prime}\right)\lambda_{\min}\left(\overline{\Sigma}\right)}\\
\le & \frac{nT\lambda_{\max}\left(\frac{\PSI^{\prime}\PSI}{T}\right)}{\lambda_{\min}\left(\widetilde{\PSI}\widetilde{\PSI}^{\prime}\right)\lambda_{\min}\left(\overline{\Sigma}\right)}.
\end{align*}
The first term is: 
\begin{align*}
 & \left\Vert A_{nT}\left\{ \Gamma\left(\hat{f}_{P}\right)-\Gamma\left(f\right)-\Gamma^{\prime}\left(\hat{f}_{P}\right)+\Gamma^{\prime}\left(f\right)\right\} \right\Vert _{L_{2}}\\
\le & \left\Vert A_{nT}\right\Vert _{L_{2}}\left\Vert \Gamma\left(\hat{f}_{P}\right)-\Gamma\left(f\right)-\Gamma^{\prime}\left(\hat{f}_{P}\right)+\Gamma^{\prime}\left(f\right)\right\Vert _{L_{2}}\\
\le & \left\Vert A_{nT}\right\Vert _{L_{2}}C_{4}\left|\hat{f}_{P}-f\right|_{s}^{2}\le C_{4}\frac{\left(nT\right)^{1/2}\left|\hat{f}_{P}-f\right|_{s}^{2}\lambda_{\max}^{1/2}\left(\frac{\PSI^{\prime}\PSI}{T}\right)}{\lambda_{\min}^{1/2}\left(\widetilde{\PSI}\widetilde{\PSI}^{\prime}\right)\lambda_{\min}^{1/2}\left(\overline{\Sigma}\right)}.
\end{align*}
The second term can be upper bounded as in Theorem \ref{Th - Asymptotic Normality},
replacing Assumption \ref{ASS4} with Assumption \ref{ASS7} (i) and
(iii).

\subsection{Estimation of the Variance Matrix}

\emph{Proof of Theorem \ref{Th - Estimation of the Variance Matrix}.}
First of all, we derive some general results about $\widehat{\overline{\Sigma}}$
that will be needed for the Wald tests. We will need to consider the
quadratic form $\PSI^{\prime}\widehat{\overline{\Sigma}}\PSI$, that
can be written as follows:

\[
\PSI^{\prime}\widehat{\overline{\Sigma}}\PSI=\PSI^{\prime}\left(\frac{1}{n}\sum_{i=1}^{n}\widehat{U}_{i}\widehat{U}_{i}^{\prime}\right)\PSI=\frac{1}{n}\sum_{i=1}^{n}\left(\PSI^{\prime}\widehat{U}_{i}\right)\left(\PSI^{\prime}\widehat{U}_{i}\right)^{\prime}.
\]
Defining $G_{\PSI}=\PSI\left(\PSI^{\prime}\PSI\right)^{-1}\PSI^{\prime}$
and $S_{\PSI}=I_{T}-G_{\PSI}$, we write $\widehat{U}_{i}$ as: 
\begin{align*}
\widehat{U}_{i}= & Y_{i}-\PSI\widehat{\beta}=Y_{i}-G_{\PSI}\overline{\mathbf{Y}}\\
= & f\left(\mathbf{X}\right)+\varepsilon_{i}-G_{\PSI}\left(f\left(\mathbf{X}\right)+\overline{\varepsilon}\right)\\
= & S_{\PSI}f\left(\mathbf{X}\right)+\left\{ \varepsilon_{i}-G_{\PSI}\overline{\varepsilon}\right\} .
\end{align*}
Therefore: 
\[
\PSI^{\prime}\widehat{U}_{i}=\PSI^{\prime}S_{\PSI}f\left(\mathbf{X}\right)+\PSI^{\prime}\left\{ \varepsilon_{i}-G_{\PSI}\overline{\varepsilon}\right\} =\PSI^{\prime}\left(\varepsilon_{i}-\overline{\varepsilon}\right)
\]
and: 
\begin{equation}
\PSI^{\prime}\left(\widehat{\overline{\Sigma}}-\overline{\Sigma}\right)\PSI=\PSI^{\prime}\left(\frac{1}{n}\sum_{i=1}^{n}\varepsilon_{i}\varepsilon_{i}^{\prime}-\overline{\varepsilon}\overline{\varepsilon}^{\prime}-\overline{\Sigma}\right)\PSI.\label{Eq - Simplified Variance Matrix}
\end{equation}
We need $\lambda_{\max}\left(\frac{1}{n}\sum_{i=1}^{n}\varepsilon_{i}\varepsilon_{i}^{\prime}-\overline{\varepsilon}\overline{\varepsilon}^{\prime}-\overline{\Sigma}\right)$
and we can majorize it as: 
\[
\lambda_{\max}\left(\frac{1}{n}\sum_{i=1}^{n}\varepsilon_{i}\varepsilon_{i}^{\prime}-\overline{\varepsilon}\overline{\varepsilon}^{\prime}-\overline{\Sigma}\right)\le\lambda_{\max}\left(\frac{1}{n}\sum_{i=1}^{n}\varepsilon_{i}\varepsilon_{i}^{\prime}-\overline{\Sigma}\right)+\lambda_{\max}\left(\overline{\varepsilon}\overline{\varepsilon}^{\prime}\right).
\]
Here
\begin{align*}
\mathbb{E}\lambda_{\max}\left(\overline{\varepsilon}\overline{\varepsilon}^{\prime}\right)= & \mathbb{E}\overline{\varepsilon}^{\prime}\overline{\varepsilon}=\sum_{t=1}^{T}\mathbb{V}\left(\frac{1}{n}\sum_{i=1}^{n}\varepsilon_{it}\right)\\
= & \frac{1}{n^{2}}\sum_{t=1}^{T}\sum_{i=1}^{n}\mathbb{V}\left(\varepsilon_{it}\right)\le\frac{T}{n}\max_{i,t}\mathbb{V}\left(\varepsilon_{it}\right)\\
\le & \frac{T}{n}\sqrt{\max_{i,t}\mathbb{E}\varepsilon_{it}^{4}}.
\end{align*}
Now we majorize the other term. In the general case, we have $\lambda_{\max}\left(\frac{1}{n}\sum_{i=1}^{n}\varepsilon_{i}\varepsilon_{i}^{\prime}-\overline{\Sigma}\right)\le\left\Vert \frac{1}{n}\sum_{i=1}^{n}\varepsilon_{i}\varepsilon_{i}^{\prime}-\overline{\Sigma}\right\Vert _{F}$.
Now:
\begin{align*}
\mathbb{E}\left\Vert \frac{1}{n}\sum_{i=1}^{n}\varepsilon_{i}\varepsilon_{i}^{\prime}-\overline{\Sigma}\right\Vert _{F}^{2}= & \mathbb{E}\left\Vert \frac{1}{n}\sum_{i=1}^{n}\left(\varepsilon_{i}\varepsilon_{i}^{\prime}-\Sigma_{i}\right)\right\Vert _{F}^{2}\\
= & \frac{1}{n^{2}}\sum_{j,k=1}^{T}\mathbb{E}\left(\sum_{i=1}^{n}\left(\varepsilon_{ij}\varepsilon_{ik}-\Sigma_{i,jk}\right)\right)^{2}\\
= & \frac{1}{n^{2}}\sum_{j,k=1}^{T}\sum_{i=1}^{n}\mathbb{E}\left(\varepsilon_{ij}\varepsilon_{ik}-\Sigma_{i,jk}\right)^{2}\\
\le & \frac{T^{2}}{n}\max_{i,t}\mathbb{E}\varepsilon_{it}^{4}
\end{align*}
where we have used the Cauchy-Schwarz' inequality to bound the terms
with $j\neq k$. Therefore:
\[
\lambda_{\max}\left(\frac{1}{n}\sum_{i=1}^{n}\varepsilon_{i}\varepsilon_{i}^{\prime}-\overline{\varepsilon}\overline{\varepsilon}^{\prime}-\overline{\Sigma}\right)=O_{\mathbb{P}}\left(\sqrt{\frac{T^{2}\max_{i,t}\mathbb{E}\varepsilon_{it}^{4}}{n}}\right).
\]
This is the bound we will use in the following proofs.

\emph{(i)} We are led to study $\frac{\overline{W}_{nT}^{\prime}\overline{W}_{nT}-R}{\sqrt{2R}}$.
We want to obtain: 
\begin{align*}
\left|\overline{W}_{nT}^{\prime}\overline{W}_{nT}-\widehat{\overline{W}}_{nT}^{\prime}\widehat{\overline{W}}_{nT}\right|= & \left|2\left(\overline{W}_{nT}-\widehat{\overline{W}}_{nT}\right)^{\prime}\overline{W}_{nT}-\left(\overline{W}_{nT}-\widehat{\overline{W}}_{nT}\right)^{\prime}\left(\overline{W}_{nT}-\widehat{\overline{W}}_{nT}\right)\right|\\
\le & 2\left\Vert \overline{W}_{nT}-\widehat{\overline{W}}_{nT}\right\Vert _{L_{2}}\left\Vert \overline{W}_{nT}\right\Vert _{L_{2}}+\left\Vert \overline{W}_{nT}-\widehat{\overline{W}}_{nT}\right\Vert _{L_{2}}^{2}.
\end{align*}
Now:
\begin{align*}
\left\Vert \overline{W}_{nT}-\widehat{\overline{W}}_{nT}\right\Vert _{L_{2}}= & \left\Vert \left(A_{nT}-\widehat{A}_{nT}\right)\Gamma^{\prime}\left(\PSI^{\prime}\PSI\right)^{-1}\PSI^{\prime}\overline{\varepsilon}\right\Vert _{L_{2}}\\
= & \left\Vert \left(I-\widehat{A}_{nT}A_{nT}^{-1}\right)A_{nT}\Gamma^{\prime}\left(\PSI^{\prime}\PSI\right)^{-1}\PSI^{\prime}\overline{\varepsilon}\right\Vert _{L_{2}}\\
\le & \left\Vert I-\widehat{A}_{nT}A_{nT}^{-1}\right\Vert _{L_{2}}\left\Vert \overline{W}_{nT}\right\Vert _{L_{2}}\\
\left|\overline{W}_{nT}^{\prime}\overline{W}_{nT}-\widehat{\overline{W}}_{nT}^{\prime}\widehat{\overline{W}}_{nT}\right|\le & \left\{ 2\left\Vert I-\widehat{A}_{nT}A_{nT}^{-1}\right\Vert _{L_{2}}+\left\Vert I-\widehat{A}_{nT}A_{nT}^{-1}\right\Vert _{L_{2}}^{2}\right\} \left\Vert \overline{W}_{nT}\right\Vert _{L_{2}}^{2}
\end{align*}
where $\left\Vert I-\widehat{A}_{nT}A_{nT}^{-1}\right\Vert _{L_{2}}=\left\Vert V_{nT}^{1/2}\widehat{V}_{nT}^{-1/2}-I\right\Vert _{L_{2}}$
and $\widehat{V}_{nT}$ is a perturbation of $V_{nT}$. Here $\left\Vert \overline{W}_{nT}\right\Vert _{L_{2}}^{2}=\overline{W}_{nT}^{\prime}\overline{W}_{nT}=O_{\mathbb{P}}\left(R\right)$.

In \citet[Th. 2]{mathias1997}, the following bound for $\left(n\times n\right)-$matrices
can be found:
\[
\left\Vert \left[H+\eta\Delta H\right]^{1/2}H^{-1/2}-I\right\Vert _{L_{2}}\le\frac{1}{2}\left(\gamma_{n}-1\right)\eta+O\left(\eta^{2}\right)
\]
where $\left\Vert H^{-1/2}\Delta H\cdot H^{-1/2}\right\Vert _{L_{2}}=1$
and $\gamma_{n}\triangleq\frac{1}{n}\sum_{j=1}^{n}\left|\cot\left(2j-1\right)\pi/2n\right|=\frac{2}{n}\sum_{j=1}^{\left[n/2\right]}\cot\left(2j-1\right)\pi/2n$.
Remark that the case $n=1$ leads directly to the majorization $\left\Vert \left[H+\eta\Delta H\right]^{1/2}H^{-1/2}-I\right\Vert _{L_{2}}=\frac{1}{2}\eta+O\left(\eta^{2}\right)$
and that $\gamma_{n}/\ln n\rightarrow2/\pi$. Alternatively, we can
write it as: 
\[
\left\Vert \left[H+\Delta H\right]^{1/2}H^{-1/2}-I\right\Vert _{L_{2}}\le\frac{1}{2}\left(\gamma_{n}-1\right)\left\Vert H^{-1/2}\Delta H\cdot H^{-1/2}\right\Vert _{L_{2}}+O\left(\left\Vert H^{-1/2}\Delta H\cdot H^{-1/2}\right\Vert _{L_{2}}^{2}\right).
\]
In our case we have: 
\[
\left\Vert H^{-1/2}\Delta H\cdot H^{-1/2}\right\Vert _{L_{2}}\le\left\Vert \Delta H\right\Vert _{L_{2}}\left\Vert H^{-1}\right\Vert _{L_{2}}\le\frac{\lambda_{\max}\left(\Delta H\right)}{\lambda_{\min}\left(H\right)}
\]
where: 
\begin{align*}
\lambda_{\min}\left(H\right)= & \lambda_{\min}\left(\frac{1}{n}\widetilde{\PSI}\left(\PSI^{\prime}\PSI\right)^{-1}\PSI^{\prime}\widehat{\overline{\Sigma}}\PSI\left(\PSI^{\prime}\PSI\right)^{-1}\widetilde{\PSI}^{\prime}\right)\\
\ge & \frac{\lambda_{\min}\left(\widetilde{\PSI}\widetilde{\PSI}^{\prime}\right)\lambda_{\min}\left(\widehat{\overline{\Sigma}}\right)}{nT\lambda_{\max}\left(\frac{\PSI^{\prime}\PSI}{T}\right)}
\end{align*}
and: 
\begin{align*}
\lambda_{\max}\left(\Delta H\right)= & \lambda_{\max}\left(\frac{1}{n}\widetilde{\PSI}\left(\PSI^{\prime}\PSI\right)^{-1}\PSI^{\prime}\left(\widehat{\overline{\Sigma}}-\overline{\Sigma}\right)\PSI\left(\PSI^{\prime}\PSI\right)^{-1}\widetilde{\PSI}^{\prime}\right)\\
\le & \frac{\lambda_{\max}\left(\widetilde{\PSI}\widetilde{\PSI}^{\prime}\right)\lambda_{\max}\left(\widehat{\overline{\Sigma}}-\overline{\Sigma}\right)}{nT\lambda_{\min}\left(\frac{\PSI^{\prime}\PSI}{T}\right)},
\end{align*}
that is: 
\[
\left\Vert H^{-1/2}\Delta H\cdot H^{-1/2}\right\Vert _{L_{2}}\le\frac{\lambda_{\max}\left(\widetilde{\PSI}\widetilde{\PSI}^{\prime}\right)}{\lambda_{\min}\left(\widetilde{\PSI}\widetilde{\PSI}^{\prime}\right)}\frac{\lambda_{\max}\left(\frac{\PSI^{\prime}\PSI}{T}\right)}{\lambda_{\min}\left(\frac{\PSI^{\prime}\PSI}{T}\right)}\frac{\lambda_{\max}\left(\widehat{\overline{\Sigma}}-\overline{\Sigma}\right)}{\lambda_{\min}\left(\widehat{\overline{\Sigma}}\right)}.
\]
The term $\frac{\lambda_{\max}\left(\frac{\PSI^{\prime}\PSI}{T}\right)}{\lambda_{\min}\left(\frac{\PSI^{\prime}\PSI}{T}\right)}$
is bounded under Assumption \ref{ASS2}. If $\frac{T\sqrt{\max_{i,t}\mathbb{E}\varepsilon_{it}^{4}}}{\lambda_{\min}\left(\overline{\Sigma}\right)\sqrt{n}}\rightarrow0$:
\begin{align*}
\lambda_{\max}\left(\widehat{\overline{\Sigma}}-\overline{\Sigma}\right)= & O_{\mathbb{P}}\left(T\sqrt{\frac{\max_{i,t}\mathbb{E}\varepsilon_{it}^{4}}{n}}\right)\\
\lambda_{\min}\left(\widehat{\overline{\Sigma}}\right)\ge & \lambda_{\min}\left(\overline{\Sigma}\right)+\lambda_{\min}\left(\widehat{\overline{\Sigma}}-\overline{\Sigma}\right)\\
= & \lambda_{\min}\left(\overline{\Sigma}\right)\cdot\left[1+\frac{\lambda_{\min}\left(\widehat{\overline{\Sigma}}-\overline{\Sigma}\right)}{\lambda_{\min}\left(\overline{\Sigma}\right)}\right]\\
= & \lambda_{\min}\left(\overline{\Sigma}\right)\cdot\left[1+O\left(\frac{\lambda_{\max}\left(\widehat{\overline{\Sigma}}-\overline{\Sigma}\right)}{\lambda_{\min}\left(\overline{\Sigma}\right)}\right)\right]\\
= & \lambda_{\min}\left(\overline{\Sigma}\right)\cdot\left[1+O_{\mathbb{P}}\left(\frac{T\sqrt{\max_{i,t}\mathbb{E}\varepsilon_{it}^{4}}}{\lambda_{\min}\left(\overline{\Sigma}\right)\sqrt{n}}\right)\right]
\end{align*}
where we have used the previously derived bound on $\lambda_{\max}\left(\widehat{\overline{\Sigma}}-\overline{\Sigma}\right)$
and the Weyl inequality for the smallest eigenvalue. Now:
\begin{align*}
\left\Vert H^{-1/2}\Delta H\cdot H^{-1/2}\right\Vert _{L_{2}}\le & O_{\mathbb{P}}\left(\frac{\lambda_{\max}\left(\widetilde{\PSI}\widetilde{\PSI}^{\prime}\right)}{\lambda_{\min}\left(\widetilde{\PSI}\widetilde{\PSI}^{\prime}\right)}\frac{T\sqrt{\max_{i,t}\mathbb{E}\varepsilon_{it}^{4}}}{\lambda_{\min}\left(\overline{\Sigma}\right)\sqrt{n}}\right)\\
\left\Vert \left[H+\Delta H\right]^{1/2}H^{-1/2}-I\right\Vert _{L_{2}}\le & O_{\mathbb{P}}\left(\ln R\left(\frac{\lambda_{\max}\left(\widetilde{\PSI}\widetilde{\PSI}^{\prime}\right)}{\lambda_{\min}\left(\widetilde{\PSI}\widetilde{\PSI}^{\prime}\right)}\frac{T\sqrt{\max_{i,t}\mathbb{E}\varepsilon_{it}^{4}}}{\lambda_{\min}\left(\overline{\Sigma}\right)\sqrt{n}}\right)\right).
\end{align*}
 Therefore, provided $\frac{\lambda_{\max}\left(\widetilde{\PSI}\widetilde{\PSI}^{\prime}\right)}{\lambda_{\min}\left(\widetilde{\PSI}\widetilde{\PSI}^{\prime}\right)}\frac{T\left(1\vee\ln R\right)\sqrt{\max_{i,t}\mathbb{E}\varepsilon_{it}^{4}}}{\sqrt{n}\lambda_{\min}\left(\overline{\Sigma}\right)}\rightarrow0$:
\begin{align*}
\frac{\left|\overline{W}_{nT}^{\prime}\overline{W}_{nT}-\widehat{\overline{W}}_{nT}^{\prime}\widehat{\overline{W}}_{nT}\right|}{\sqrt{2R}}\le & \frac{\left\{ 2\left\Vert I-\widehat{A}_{nT}A_{nT}^{-1}\right\Vert _{L_{2}}+\left\Vert I-\widehat{A}_{nT}A_{nT}^{-1}\right\Vert _{L_{2}}^{2}\right\} \left\Vert \overline{W}_{nT}\right\Vert _{L_{2}}^{2}}{\sqrt{2R}}\\
= & O_{\mathbb{P}}\left(\frac{\lambda_{\max}\left(\widetilde{\PSI}\widetilde{\PSI}^{\prime}\right)}{\lambda_{\min}\left(\widetilde{\PSI}\widetilde{\PSI}^{\prime}\right)}\frac{T\sqrt{R}\left(1\vee\ln R\right)\sqrt{\max_{i,t}\mathbb{E}\varepsilon_{it}^{4}}}{\lambda_{\min}\left(\overline{\Sigma}\right)\sqrt{n}}\right).
\end{align*}
The result of the theorem follows.

\end{document}